\documentclass[12pt,onecolumn]{IEEEtran}
\usepackage{epsf}
\usepackage{graphicx}% Extended graphics
\usepackage{amsmath} \usepackage{amssymb}
\usepackage{booktabs}
\usepackage{subfigure}
\usepackage{cases}
\usepackage{color}
\newenvironment{proof}{\begin{IEEEproof}}{\end{IEEEproof}}

\newtheorem{theorem}{Theorem}[section]
\newtheorem{definition}{Definition}[section]
\newtheorem{proposition}{Proposition}[section]
\newtheorem{lemma}{Lemma}[section]

\newtheorem{example}{Example}[section]

\long\def\symbolfootnote[#1]#2{\begingroup%
\def\thefootnote{\fnsymbol{footnote}}\footnote[#1]{#2}\endgroup}

\def\dref#1{(\ref{#1})}

 \def\Cap{\bigcap} \def\Cup{\bigcup}

\def\be{\begin{equation}} \def\ee{\end{equation}}
\def\ba{\begin{array}} \def\ea{\end{array}} \def\bna{\begin{eqnarray}}
\def\ena{\end{eqnarray}}

 \def\bna{\begin{eqnarray}}
\def\ena{\end{eqnarray}} \def\dref#1{(\ref{#1})}
\begin{document}

\title{Improving on the Cut-Set Bound via Geometric Analysis of Typical Sets}

%\author{\IEEEauthorblockN{Xiugang Wu and Liang-Liang Xie}\\
%\IEEEauthorblockA{
%University of Waterloo, Waterloo, ON, Canada N2L 3G1 \\
%Email: x23wu@uwaterloo.ca, llxie@uwaterloo.ca} }

\author{Xiugang Wu, Ayfer \"{O}zg\"{u}r and
      Liang-Liang Xie

\thanks{This work was supported in part by the NSF CAREER award 1254786 and by the Center for Science of Information (CSoI), an NSF Science and Technology Center, under grant agreement CCF-0939370. This paper was presented in part at the 2015 IEEE International Symposium on Information Theory \cite{ISIT2015} and the 2016 International Zurich Seminar on Communications \cite{IZS2016}.}
\thanks{X. Wu and A. \"{O}zg\"{u}r are with the Department of Electrical Engineering, Stanford University, Stanford, CA 94305, USA (e-mail: x23wu@stanford.edu; aozgur@stanford.edu).}
\thanks{L.-L. Xie is with the Department
of Electrical and Computer Engineering, University of Waterloo, Waterloo, Ontario, Canada N2L 3G1 (e-mail: llxie@uwaterloo.ca).}
}

\maketitle

\begin{abstract}
We consider the discrete memoryless symmetric primitive relay channel, where, a source $X$ wants to send information to a destination $Y$ with the help of a relay $Z$ and the relay can communicate to the destination via an error-free digital link of rate $R_0$, while $Y$ and  $Z$ are conditionally independent and identically distributed given $X$. We develop two new upper bounds on the capacity of this channel that are tighter than existing bounds, including the celebrated cut-set bound. Our approach significantly deviates from the standard information-theoretic approach for proving upper bounds on the capacity of multi-user channels. We build on the blowing-up lemma to analyze the probabilistic geometric relations between the typical sets of the  $n$-letter random variables associated with a reliable code for communicating over this channel. These relations translate to new entropy inequalities between the $n$-letter random variables involved.

As an application of our bounds, we study an open question posed by (Cover, 1987), namely, what is the minimum needed $Z$-$Y$ link rate $R_0^*$ in order for the capacity of the relay channel to be equal to that of the broadcast cut. We consider the special case when the $X$-$Y$ and $X$-$Z$ links are both binary symmetric channels. Our tighter bounds on the capacity of the relay channel immediately translate to tighter lower bounds for $R_0^*$. More interestingly, we show that when $p\to 1/2$, $R_0^*\geq 0.1803$; even though the broadcast channel becomes completely noisy as $p\to 1/2$ and its capacity, and therefore the capacity of the relay channel, goes to zero, a strictly positive rate $R_0$ is required for the relay channel capacity to be equal to the broadcast bound. Existing upper bounds on the capacity of the relay channel, and the cut-set bound in particular, would rather imply $R_0^*\to 0$, while achievability schemes require $R_0^*\to 1$. We conjecture that $R_0^*\to 1$ as $p\to 1/2$.

\end{abstract}

\section{Introduction}\label{S:Introduction}

Characterizing the capacity of relay channels \cite{van71} has been a long-standing open problem in network information theory. The seminal work of Cover and El Gamal \cite{covelg79} has introduced two basic achievability schemes: Decode-and-Forward and Compress-and-Forward, and derived a general upper bound on the capacity of this channel, now known as the cut-set bound. Over the last decade, significant progress has been made on the achievability side: these schemes have been extended and unified to multi-relay networks \cite{XieKumar05}--\cite{WuXie14} and many  new relaying strategies have been discovered, such as Amplify-and-Forward, Quantize-Map-and-Forward, Compute-and-Forward, Noisy Network Coding, Hybrid Coding etc \cite{schein}--\cite{Hybrid}.  However, the progress on developing upper bounds that are tighter than the cut-set bound has been relatively limited. In particular, in most of the special cases where the capacity is known, the converse is given by the cut-set bound \cite{covelg79}, \cite{orthogonal}--\cite{deterministic}.

In general, however, the cut-set bound is known to be not tight. Specifically, consider the primitive relay channel depicted in Fig. \ref{F:primitive}, where the source's input $X $ is received by the relay $Z$ and the destination $Y$ through a channel  $p(y,z|x)$, and the relay $Z$ can communicate to the destination $Y$ via an error-free digital link of rate $R_0$. When $Y$ and $Z$ are conditionally independent given $X$, and $Y$ is a stochastically degraded version of $Z$, Zhang \cite{Zhang} uses the blowing-up lemma \cite{blowingup} to show that
the capacity can be strictly smaller than the cut-set bound in certain regimes of this channel.
%the inequality between the capacity and the cut-set bound is indeed strict in certain regimes of this channel. In other words, the capacity can indeed be strictly smaller than the cut-set bound.
However, Zhang's result does not provide any information regarding the gap or suggest a way to compute it. For a special case of the primitive relay channel where the noise for the $X$-$Y$ link is modulo additive and $Z$ is a corrupted version of this noise,  Aleksic, Razaghi and Yu characterize the capacity and show that it is strictly lower than the cut-set bound \cite{mod}. %This result has been later extended by Tandon and Ulukus \cite{TandonUlukus} to the case when $Z$ is independent of $X$ and acts as the state of the $X$-$Y$ link. However,
While this result provides an exact capacity characterization for a non-trivial special case, it builds strongly on the peculiarity of the channel model and in this respect its scope is more limited than Zhang's result.

\begin{figure}[hbt]
\centering
\includegraphics[width=0.3\textwidth]{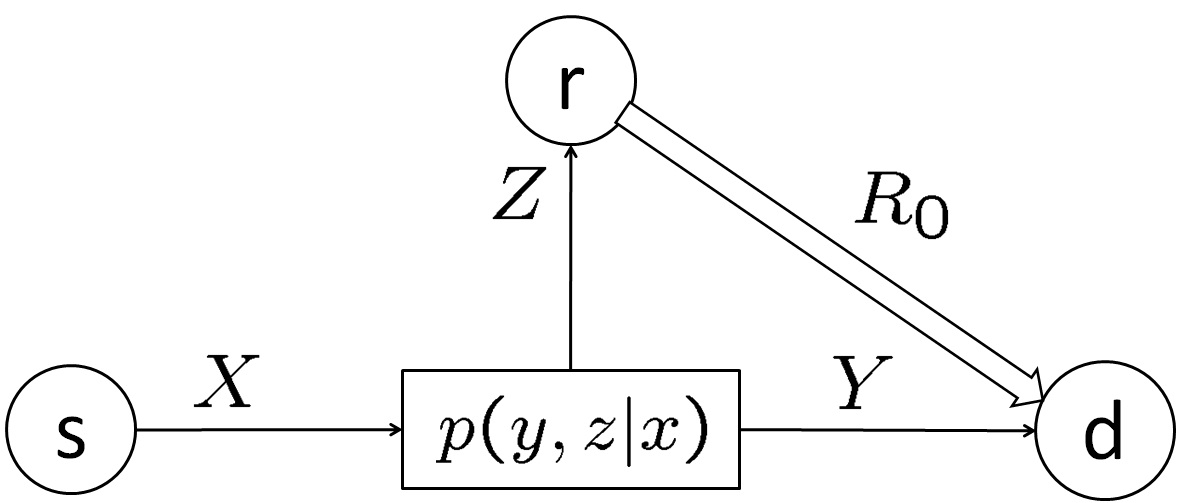}
\caption{Primitive relay channel.}
 \label{F:primitive}
\end{figure}

More recently, a new upper bound demonstrating an explicit gap to the cut-set bound  was developed by Xue \cite{Xue} for general primitive relay channels. %This new bound was first established for the symmetric case, i.e., when $Y$ and $Z$ are conditionally independent and identically distributed given $X$, and then extended to the general asymmetric case employing channel simulation theory \cite{sim1}--\cite{sim2}.
Xue's bound relates the gap of the cut-set bound to the reliability function of the  $X$-$Y$ link. In particular, it builds on the blowing-up lemma to lower bound the successful decoding probability based on $Y$ only and then compare it to the reliability function of the single-user channel $X$-$Y$. Unlike Zhang's result, Xue's bound identifies an explicit gap to the cut-set bound that can be numerically computed.
%Given a reliable code for the relay channel, it extracts a subcodebook, and building on the existence of this subcodebook and the blowing-up lemma, it lower bounds the successful decoding probability  when only $Y$ is employed for decoding at the destination. It then compares this decoding probability, based only on the output $Y$, to the reliability function of the single-user channel $X$-$Y$.
However, it also has some obvious drawbacks over the cut-set bound. For example, it bounds only the information flow from the source and the relay to the destination (the multiple-access cut) and ignores the flow from the source to the relay and the destination (the broadcast cut). As such, it can be also looser than the cut-set bound since it does not capture the inherent trade-off between maximizing the information flows across these two different cuts of the network.

In this paper, we present two new upper bounds on the capacity of the primitive relay channel that are generally tighter than the cut-set bound. To simplify exposition, we concentrate on the symmetric case ($Y$ and $Z$ are conditionally independent and identically distributed given $X$) in this paper,  however our results can be extended to the asymmetric case via channel simulation arguments \cite{ISIT2016}. Just like Zhang and Xue, we critically build on the blowing up lemma, however we develop novel ways for utilizing it which lead to simpler arguments and tighter results. In general, proving an upper bound on the capacity of a multi-user channel involves dealing with entropy relations between the various $n$-letter random variables induced by the reliable code and the channel structure (together with using Fano's inequality). In order to prove the desired relations between the entropies of the $n$-letter random variables involved, in this paper we consider their $B$-length i.i.d. extensions (leading to length $B$ i.i.d. sequences of $n$-letter random variables). We then use the blowing-up lemma to analyze the geometry of the typical sets associated with these  $B$-length sequences. The key step in our development is to translate the (probabilistic) geometric relations between these typical sets into new entropy relations between the random variables involved. While both of our bounds are based on this same approach, they use different arguments to translate the geometry of the typical sets to entropy relations, and eventually lead to two different bounds on the capacity of the channel which do not include each other in general.

As an application of our bounds, we consider the binary symmetric channel, i.e., we assume both the $X$-$Y$ and $X$-$Z$ links are binary symmetric channels with crossover probability, say, $p$. We demonstrate that both our bounds perform strictly better than the cut-set bound and Xue's bound, and particularly, our second bound provides  considerable gain over these earlier bounds. We then use our bounds to investigate an open question posed by Cover \cite{Cover'sopenproblem} which asks for the  minimum required $Z$-$Y$ link rate $R^*_0$ in order for the capacity of the relay channel to be equal to the capacity of the broadcast cut, i.e. $\max_{p(x)}I(X;Y,Z)$. Obviously as $R^*_0$ becomes larger the capacity of the relay channel does approach the capacity of the broadcast cut. For example, in the binary symmetric case if $R_0=1$, the relay can convey its noisy observation as it is to the destinaton, therefore the broadcast cut capacity is trivially achievable. In this sense, Cover's open problem asks how smaller $R_0$ can be made than $1$ without decreasing the capacity of the relay channel. Interestingly, there is a striking dichotomy between the currently available upper and lower bounds for $R^*_0$ when $p\to 1/2$, i.e. when the broadcast channel becomes completely noisy and its capacity, and therefore the capacity of the relay channel, goes to zero. Achievability schemes, Hash-and-Forward in particular, require $R_0\to 1$ even though the capacity itself tends to zero. The cut-set bound and Xue's bound, on the other hand, require $R_0\to 0$ in order for the capacity to be equal to the broadcast capacity. By strengthening our second bound in this specific case, we show that $R^*_0\geq 0.1803$; indeed a strictly positive rate $R_0$ is needed in order to achieve the vanishing broadcast capacity. We conjecture that $R^*_0\to 1$ when $p\to 1/2$; to achieve the broadcast capacity the relay has no choice but to forward its observation, which is almost pure noise, as it is to the destination.

\subsection{Organization of the Paper}
The remainder of the paper is organized as follows. Sections \ref{S:ChannelModel} and \ref{S:existing} introduces the channel model and reviews the existing upper bounds for primitive relay channels, respectively.  Section \ref{S:Newbounds} discusses our new upper bounds for symmetric primitive relay channels in detail,
followed by a treatment on the binary symmetric channel case in Section \ref{S:BSC}. Sections \ref{S:sharpenproof}, \ref{S:novelproof} and \ref{S:furtherimprovement} are then dedicated to the proofs of our bounds. Finally, some concluding remarks are included in Section \ref{conclusion}.

\section{Channel Model}\label{S:ChannelModel}

Consider a primitive relay channel as depicted in Fig. \ref{F:primitive}. The source's input $X $ is received by the relay $Z$ and the destination $Y$ through a channel
$$(\Omega_X, p(y,z|x),  \Omega_Y \times \Omega_Z)$$
where $\Omega_X ,\Omega_Y $ and $\Omega_Z $ are finite sets denoting the alphabets of the source, the destination and the relay, respectively, and $p(y,z|x)$ is the channel transition probability; the relay $Z$ can communicate to the destination $Y$ via an error-free digital link of rate $R_0$.

For this channel, a code of rate $R$ for $n$ channel uses, denoted by $$(\mathcal{C}_{(n,R)}, f_n(z^n), g_n(y^n,f_n(z^n))), \mbox{ or simply, } (\mathcal{C}_{(n,R)}, f_n, g_n), $$
consists of the following:
\begin{enumerate}
  \item A codebook at the source $X$,
$$\mathcal{C}_{(n,R)}=\{x^n (m)\in\Omega_X^n, m\in \{1,2,\ldots, 2^{nR}\} \};$$
  \item An encoding function at the relay $Z$,
$$f_n: \Omega_Z^n \rightarrow \{1,2,\ldots, 2^{nR_0}\};$$
  \item A decoding function at the destination $Y$,
$$g_n: \Omega_Y^n  \times \{1,2,\ldots, 2^{nR_0}\}  \rightarrow \{1,2,\ldots, 2^{nR}\}.$$
\end{enumerate}

The average probability of error of the code is defined as
$$P_e^{(n)}=\mbox{Pr}(g_n(Y^n,f_n(Z^n)) \neq M ),$$
where the message $M$ is assumed to be uniformly drawn from the message set $ \{1,2,\ldots, 2^{nR}\}$. A rate $R$ is said to be achievable if there exists a sequence of codes
$$\{(\mathcal{C}_{(n,R)}, f_n, g_n)\}_{n=1}^{\infty}$$
such that the average probability of error $P_e^{(n)} \to 0$ as $n \to \infty$.

The capacity of the primitive relay channel is the supremum of all achievable rates, denoted by $C(R_0)$. Also, denote by $C_{XY},C_{XZ}$ and $C_{XYZ}$ the capacities of the channels
$X$-$Y$, $X$-$Z$, and $X$-$YZ$, respectively. Obviously, we have $C(0)=C_{XY}$ and $C(\infty)=C_{XYZ}$.

\subsection{Symmetric Primitive Relay Channel}  In this paper, we focus on the symmetric case of the primitive relay channel, that is, when $Y$ and $Z$ are conditionally independent and identically distributed given $X$. Formally, a primitive relay channel is said to be symmetric if
\begin{enumerate}
  \item $p(y,z|x)=p(y|x)p(z|x)$,
  \item $\Omega_Y = \Omega_Z:=\Omega$, and $\mbox{Pr}(Y=\omega|X=x)=\mbox{Pr}(Z=\omega|X=x)$ for any $\omega \in \Omega$ and $x \in \Omega_X$.
\end{enumerate}%
%i) $p_{_{  Y,Z|X} }(y,z|x)=p_{_{ Y|X}}(y|x)p_{_{Z|X}}(z|x)$, and ii) $\Omega_Y = \Omega_Z:=\Omega$, and $p_{_{Y|X}}(\omega|x)=p_{_{ Z|X}}(\omega|x)$ for any $\omega \in \Omega$ and $x \in \Omega_X$.
In this case, we also use $p(\omega|x)$ to denote the transition probability of both the $X$-$Y$ and $X$-$Z$ channels.

\section{Existing Upper Bounds for Primitive Relay Channels}\label{S:existing}

For general primitive relay channels, the well-known cut-set bound can be stated as follows.
\begin{proposition}[Cut-set Bound]\label{P:cutset}
For the general primitive relay channel, if a rate $R$ is achievable, then there exists some $p(x)$ such that
\begin{numcases}{}
 R   \leq I(X;Y,Z)\label{E: cut1} \\
R    \leq   I(X;Y)+R_0  \label{E: cut2}.
\end{numcases}
\end{proposition}

Inequalities \dref{E: cut1} and \dref{E: cut2} are generally known as the broadcast bound and multiple-access bound, since they correspond to the broadcast channel $X$-$YZ$ and
multiple-access channel $XZ$-$Y$, respectively.

Note that although the cut-set bound in \dref{E: cut1}--\dref{E: cut2} is tight for most of the cases where the capacity is determined  \cite{covelg79}, \cite{orthogonal}--\cite{deterministic}, it is known to be not tight in general.  The first counterexample was given by Zhang in \cite{Zhang}, where he considered a class of stochastically degraded primitive relay channels. Using the blowing-up lemma \cite{blowingup}, he showed that the capacity of the channel can be strictly smaller than the cut-set bound.

%There have been some results demonstrating the looseness of the bound in \dref{E: cut1}--\dref{E: cut2}. In particular, Zhang \cite{Zhang} considered a class of stochastically degraded primitive relay channels, and showed that the capacity can be strictly smaller than the cut-set bound.

%Zhang in \cite{Zhang} provided the first result indicating that the above bound could be loose.

\begin{proposition}[Zhang \cite{Zhang}]
\label{P:Zhang}
For a primitive relay channel, if $Y$ and $Z$ are conditionally independent given $X$, and $Y$ is a stochastically degraded version of $Z$ (i.e. there exists some $q(y|z)$ such that $p(y|x)=\sum_{z}p(z|x)q(y|z)$),  then the capacity $C(R_0)$ of the channel satisfies
\begin{align}
C(R_0)< C_{XY}+R_0 \label{E:zhang}
\end{align}
when
\begin{align}
R_0>\max_{p(x):I(X;Y)=C_{XY}} I(X;Z)-C_{XY}. \label{E:R_0}
\end{align}
\end{proposition}

In the regime where $R_0$ satisfies both \dref{E:R_0} and the condition
\begin{align*}
C_{XY}+ R_0 < \max_{p(x):I(X;Y)=C_{XY}} I(X;Y,Z),
\end{align*}
the cut-set bound becomes $C_{XY}+R_0$, however the strictness of the inequality in \dref{E:zhang} implies that the cut-set bound  is loose with some positive gap. Roughly speaking this corresponds to the regime where the cut-set bound is limited by the multiple-access bound but the source-relay channel is not strong enough to enable the relay to trivially decode the transmitted message. However, note that Zhang's result does not provide any information about how large the gap can be.

Recently, a new upper bound demonstrating an explicit gap to the cut-set bound  was developed by Xue \cite{Xue}.  This new bound was first established for the symmetric case, and then extended to the general asymmetric case employing channel simulation theory \cite{sim1}--\cite{sim2}. The proof uses a generalized version of the blowing-up lemma \cite{Xue} to characterize the successful decoding probability based only on $Y$ and then compares it with the reliability function for the single-user channel $X$-$Y$.  Xue's bound specialized for the symmetric case is given as follows.
\begin{proposition}[Xue's Bound]
\label{P:Xue}
For the symmetric primitive relay channel, if a rate $R$ is achievable, then there exists some $a \in [0, R_0]$ such that
\begin{numcases}{}
 R\leq C_{XY} + R_0  - a  \label{E:Xue1} \\
 E(R)     \leq  H(\sqrt a)+ \sqrt a \log |\Omega| \label{E:Xue2}
\end{numcases}
where $H(r)$ is the binary entropy function, and $E(R)$ is the reliability function for the $X$-$Y$ link defined as
\begin{align}
E(R):=\max_{\rho \in [-1,0)} (-\rho R+\min_{p(x)}E_0(\rho, p(x)) ) 
\end{align}
with
\begin{align*}
E_0 (\rho, p(x)):= -\log \left[\sum_{y} \left( \sum_x p(x)p(y|x)^{\frac{1}{1+\rho}}  \right)^{1+\rho}   \right].
\end{align*}
\end{proposition}

It can be seen that Xue's bound modifies the original multiple-access bound \dref{E: cut2} by introducing an additional term ``$-a$'' in  \dref{E:Xue1},
where ``$a$'' is a non-negative auxiliary variable subject to the constraint \dref{E:Xue2}. As noted in \cite{Xue}, this implies that the capacity of the symmetric primitive relay channel is  strictly less than $C_{XY}+R_0$ for any $R_0>0$. To see this, consider any rate $R>C_{XY}$. Then it follows from \cite{Arimoto} that $E(R)>0$, which forces $a$ to be strictly positive in light of \dref{E:Xue2}, and thus  $R<C_{XY}+R_0$ by \dref{E:Xue1}. Since $a$ here is numerically computable, Xue's bound in fact improves over Zhang's result in the sense that it provides a lower bound to the gap of the cut-set bound. %, where the improvement comes from the stronger argument using the generalized blowing-up lemma as aforementioned.

%To understand Xue's bound, one can interpret the ``$a$'' in \dref{IE:Xue1} as the normalized uncertainty of the relay's transmission, say $I_n$, given the source's input $X^n$, i.e., $\frac{1}{n}H(I_n|X^n)$. Suppose a rate $R$ is close to $C_{XY} + R_0$, i.e., $R\approx C_{XY} + R_0$, then both the $X$-$Y$ and $Z$-$Y$ links must convey (almost) noiseless information to the destination. This would require $I_n$ to be a (almost) deterministic function of $X^n$, i.e., $a\approx 0$, even though the $X$-$Z$ link may be random.
%However, since $Y$ and $Z$ are symmetric, it can be shown using the generalized blowing-up lemma that given $a\approx 0$, the destination has a good chance to correctly guess the relay's transmission $I_n$ based on its own observation $Y^n$ and then successfully decode. Since such a decoding method is solely based on $Y^n$, the successful decoding probability is subject to the universal bound of the reliability function for the $X$-$Y$ channel, leading us to  the constraint on $a$ in \dref{IE:Xue2}, which bounds $a$ away from 0 and thereby bounds the rate $R$ away from  $C_{XY} + R_0$.

Nevertheless, Xue's bound also has two obvious drawbacks: i) compared to the cut-set bound, it lacks a constraint on the broadcast cut and therefore decouples the information flow over the broadcast and multiple-access cuts of the channel (note that his result can be always amended by including the bound $R\leq C_{XYZ}$, however with such an amendment this bound would have no coupling with those in \eqref{E:Xue1} and \eqref{E:Xue2}, which can be potentially coupled through $p(x)$ as done in the cut-set bound);
%Because of this, Xue's bound can be looser than the cut-set bound since it does not capture the inherent trade-off between maximizing the information flows across these two different cuts of the network.
ii) there is no coupling between \eqref{E:Xue1} and \eqref{E:Xue2} which can also  benefit from a coupling through the input distribution $p(x)$. Our bounds presented in the next section overcome these drawbacks  and further improve on  Xue's bound. They are also structurally different from Xue's bound as they  involve only basic information measures and do not involve the reliability function.

%\begin{remark}
%Compared to the cut-set bound, however, Xue's bound is independent of the input distribution $p(x)$, resulting in the following two drawbacks: i) it lacks
%the broadcast bound \dref{E: cut1} as in the cut-set bound\footnote{An input distribution independent bound $R\leq C_{XYZ}$ always holds, nevertheless.};
%ii) it loosens the input distribution dependent term ``$I(X;Y)$'' in \dref{E: cut2} to ``$C_{XY}$'' in \dref{E:Xue1}. As we will show next, these drawbacks can be indeed overcome in our sharpened
% bound, where some other improvements are also made.
%\end{remark}

\section{New Upper Bounds for Symmetric Primitive Relay Channels}\label{S:Newbounds}
This section presents two new upper bounds on the capacity of symmetric primitive relay channels that are generally tighter than the cut-set bound. Before stating our main theorems, in the following subsection we first explain the relation of our new bounds to the cut-set bound.

\subsection{Improving on the Cut-Set Bound}\label{SS:improvingcutset}
Let the relay's transmission be denoted by $I_n=f_n(Z^n)$. Let us recall the derivation of the cut-set bound. The first step in deriving \dref{E: cut1}--\dref{E: cut2} is to use Fano's inequality to conclude that
$$
nR\leq I(X^n;Y^n,I_n) +n\epsilon.
$$
We can then either proceed as
\begin{align*}
nR&\leq I(X^n;Y^n,I_n) +n\epsilon\\
&\leq I(X^n;Y^n,Z^n) +n\epsilon\\
&\leq nI(X;Y,Z) +n\epsilon
\end{align*}
to obtain the broadcast bound \dref{E: cut1}, where the second inequality follows from the data processing inequality and the single letterization in the third line can be either done with a time-sharing or fixed composition code argument\footnote{Note that the time-sharing or the fixed composition code argument for single letterization is needed to preserve the coupling to the second inequality in \eqref{eq:int2} via $X$.}; or we can proceed  as
\begin{align}
nR&\leq I(X^n;Y^n,I_n) +n\epsilon\nonumber\\
&\leq I(X^n;Y^n)+H(I_n|Y^n)-H(I_n|X^n) +n\epsilon\label{eq:int}\\
&\leq nI(X;Y)+nR_0 +n\epsilon\label{eq:int2}
\end{align}
to obtain the multiple-access bound \dref{E: cut2}, where to obtain the last inequality we upper bound $H(I_n|Y^n)$ by $nR_0$ and use the fact that $H(I_n|X^n)$ is non-negative.
%is to upper bound the $n$-letter mutual information $I(X^n;Y^n,I_n)$ with single-letter expressions. In particular, the broadcast bound \dref{E: cut1} follows from upper bounding $I(X^n;Y^n,I_n)$ by $I(X^n;Y^n,Z^n)$, which is then bounded by $nI(X;Y,Z)$; and the multiple-access bound \dref{E: cut2} is obtained by expanding $I(X^n;Y^n,I_n)$ as $I(X^n;Y^n)+H(I_n|Y^n)-H(I_n|X^n)$, and upper bounding $I(X^n;Y^n)$ and $H(I_n|Y^n)$ with $nI(X;Y)$ and $nR_0$ respectively while lower bounding $H(I_n|X^n)$ by $0$.

Instead of simply lower bounding $H(I_n|X^n)$ by $0$ in the last step, our bounds presented in the next two subsections are based on letting $H(I_n|X^n)=na_n$ and proving a third inequality that forces $a_n$ to be strictly non-zero. This new inequality is based on capturing the tension between how large $H(I_n|Y^n)$ and how small $H(I_n|X^n)$ can be. Intuitively, it is easy to see this tension. Specifically, suppose $H(I_n|X^n)\approx 0$, then roughly speaking, this implies that given the transmitted codeword $X^n$, there is no ambiguity about $I_n$, or equivalently, all the $Z^n$ sequences jointly typical with $X^n$ are mapped to the same $I_n$. Since $Y^n$ and $Z^n$ are statistically equivalent given $X^n$ (they share the same typical set given $X^n$) this would further imply that $I_n$ can be determined based on $Y^n$, and therefore $H(I_n|Y^n)\approx 0$. This would force the rate to be even smaller than $I(X;Y)$.

Equivalently, rewriting \eqref{eq:int} and \eqref{eq:int2} as
\begin{equation}\label{eq:difI}
R \leq nI(X;Y)+I(I_n;X^n)-I(I_n;Y^n) +n\epsilon,
\end{equation}
our approach can be thought of as fixing the first $n$-letter mutual information to be $I(I_n;X^n)\leq n(R_0-a_n)$ and controlling the second $n$-letter mutual information. In doing so, we only build on the Markov chain structure $I_n \leftrightarrow Z^n \leftrightarrow X^n \leftrightarrow Y^n$  and the fact that $Z^n$ and $Y^n$ are conditionally i.i.d. given $X^n$. In particular, we do not employ the fact that these random variables are associated with a reliable code. Note that this approach of directly studying the relation between the $n$-letter information measures involved significantly deviates from the standard approach in network information theory for developing converses, where one usually seeks to single letterize such $n$-letter expressions.
 %In general, there is a trade-off between how close the rate can get to the multiple-access bound $I(X;Y)+R_0$ and how much it can exceed the point-to-point capacity $I(X;Y)$ of the $X$-$Y$ link.

More precisely, we proceed as follows. We fix $H(I_n|X^n)=na_n$ and leave this term as it is in \eqref{eq:int}, yielding
$$
R\leq I(X;Y)+R_0 -a_n+\epsilon.
$$
We then prove the following two upper bounds on $I(I_n;X^n)-I(I_n;Y^n)$ in terms of $a_n$:
\begin{equation}\label{eq:diffI1}
I(I_n;X^n)-I(I_n;Y^n)\leq n \left[H\left(\sqrt{\frac{a_n \ln 2}{2}}\right) + \sqrt{\frac{a_n \ln 2}{2}}\log (|\Omega|-1)-a_n\right], 
\end{equation}
where $H(r)$ is the binary entropy function; and 
\begin{equation}\label{eq:diffI2}
I(I_n;X^n)-I(I_n;Y^n)\leq n  \Delta\left(p(x),\sqrt{\frac{a_n \ln 2}{2}}\right)  ,
\end{equation}
where $\Delta\left(p(x),\sqrt{\frac{a_n \ln 2}{2}}\right)$ is a quantity that depends on the input distribution $p(x)$ and $a_n$, which we will formally define in Section~\ref{SS:Newbound2}. These two bounds are obtained via bounding $H(I_n|Y^n)$ and $H(Y^n|I_n)$ in terms of $a_n$ respectively, and combined with \eqref{eq:difI} they immediately yield new constraints on $R$.

%The first bound is obtained by focusing on bounding $H(I_n|Y^n)$, while the second bound is obtained by bounding $H(Y^n|I_n)$ in terms of $a_n$. These bounds combined with \eqref{eq:difI} immediately yield new constraints on $R$.  
%This new constraint is obtained by writing
%\begin{equation}\label{eq:thirdineq}
%nR\leq I(X^n;Y^n,I_n)+n\epsilon=H(Y^n,I_n)-H(Y^n|X^n)-H(I_n|X^n)+n\epsilon,
%\end{equation}
%and upper bounding $H(Y^n,I_n)$ in terms of $a_n$. We do this in two different ways corresponding to the two different ways of expanding  $H(Y^n,I_n)$, i.e.
%\begin{align*}
%H(Y^n,I_n)&=H(Y^n)+H(I_n|Y^n)\\
%&=H(I_n)+H(Y^n|I_n).
%\end{align*}
%Our first bound attacks the first conditional entropy term and is based on proving that
%$$
%H(I_n|Y^n)\leq n \left[ H\left(\sqrt{\frac{a_n \ln 2}{2}}\right) + \sqrt{\frac{a_n \ln 2}{2}}\log (|\Omega|-1)\right],
%$$
%where $H(r)$ is the binary entropy function. Our second bound attacks the second conditional entropy term and is based on proving that
%$$
%H(Y^n|I_n)\leq H(X^n|I_n)-H(X^n|Z^n)+n\left[H(Y|X)+\Delta\left(p(x),\sqrt{\frac{a_n \ln 2}{2}}\right)\right],
%$$
%where $\Delta\left(p(x),\sqrt{\frac{a_n \ln 2}{2}}\right)$ is a quantity that depends on the input distribution $p(x)$ and $a_n$, which we formally define in Section~\ref{SS:Newbound2}. Once these entropy relations are proved, it is easy to plug them in \eqref{eq:thirdineq} and see how they lead to the theorems stated in the next two sections. 

The heart of our argument is therefore to prove the two bounds in \eqref{eq:diffI1} and \eqref{eq:diffI2}. To accomplish this, we suggest a new set of proof techniques. In particular, we look at the $B$-letter i.i.d. extensions of the random variables $X^n, Y^n, Z^n$ and $I_n$ and study the geometric relations between their typical sets by using the generalized blowing-up lemma. While we use this same general approach for obtaining \eqref{eq:diffI1} and \eqref{eq:diffI2}, we build on different arguments in each case, which eventually leads to two different bounds on the capacity of the relay channel that do not include each other in general.

%It should be pointed out that although both of our bounds involve the same set of techniques, their proofs are essentially different. In particular, note that $H(Y^n,I_n)$ can be written as $H(Y^n)+H(I_n|Y^n)$ or $H(I_n)+H(Y^n|I_n)$; our two bounds will target at bounding the conditional entropies $H(I_n|Y^n)$ and $H(I_n|Y^n)$ respectively with different arguments.

\subsection{Via Bounding $H(I_n|Y^n)$}\label{SS:Newbounds}
Our first bound builds on bounding $H(I_n|Y^n)$ and it is given by the following theorem.

\begin{theorem}\label{T:sharpen}
For the symmetric primitive relay channel, if a rate $R$ is achievable, then there exists some $p(x)$  and
\begin{align}\label{E:newconstraint_a}
a \in \left[0, \min\left\{R_0, H(Z|X), \frac{2}{\ln 2 } \left( \frac{|\Omega|-1}{|\Omega|} \right)^2\right\} \right]
\end{align}
such that
\begin{numcases}{}
R \leq I(X;Y,Z)   \label{E:sharpened1}\\
 R \leq I(X;Y) + R_0- a \label{E:sharpened2} \\
R  \leq I(X;Y) + H\left(\sqrt{\frac{a \ln 2}{2}}\right) + \sqrt{\frac{a \ln 2}{2}}\log (|\Omega|-1) -a.
\label{E:sharpened3}
\end{numcases}
\end{theorem}

Clearly our bound in Theorem \ref{T:sharpen}  implies the cut-set bound in Proposition \ref{P:cutset}. In fact, it can be checked that our bound is \emph{strictly} tighter than the cut-set bound for any $R_0 >0$. For this, note that \dref{E:sharpened2} will reduce to \dref{E: cut2} only if $a=0$; however, if $a=0$ then \dref{E:sharpened3} will constrain $R$ by the rate $I(X;Y)$ which is lower than the cut-set bound.

Our bound is also generally tighter than Xue's bound, and since Xue's bound implies Zhang's result \cite{Zhang}, so does our bound.
In particular, our bound overcomes the drawbacks of Xue's bound that are observed in Section \ref{S:existing}, and furthermore tightens the constraint \dref{E:Xue2} on $a$  to \dref{E:sharpened3}. By contrasting \dref{E:sharpened1}--\dref{E:sharpened3} to \dref{E:Xue1}--\dref{E:Xue2}, we note the following improvements:
\begin{enumerate}
  \item Our bound introduces the missing constraint on the broadcast cut \dref{E:sharpened1} and couples it with \dref{E:sharpened2}--\dref{E:sharpened3} through the input distribution $p(x)$.
  \item The term $C_{XY}$ in  \dref{E:Xue1} is replaced by $I(X;Y)$ in \dref{E:sharpened2}. Since the distribution $p(x)$ in Theorem \ref{T:sharpen} has to be chosen to satisfy all the constraints \dref{E:sharpened1}--\dref{E:sharpened3},  it may not necessarily maximize $I(X;Y)$, and thus \dref{E:sharpened2} is in general stricter than \dref{E:Xue1}.
  \item The constraint \dref{E:Xue2} on $a$ is replaced by \dref{E:sharpened3}. To show that the latter is stricter, rewrite it as
  \begin{align}\label{E:rewrite}
  R -I(X;Y) \leq H\left(\sqrt{\frac{a \ln 2}{2}}\right) + \sqrt{\frac{a \ln 2}{2}}\log (|\Omega|-1) -a.
  \end{align}
Note that  \dref{E:Xue2}  is active only if $R>C_{XY}$. In Appendix \ref{A:reliability} we show that in this case the L.H.S. (left-hand-side) of \dref{E:rewrite} is generally greater than that of \dref{E:Xue2}, while the the R.H.S. (right-hand-side) of \dref{E:rewrite} is obviously less than that of \dref{E:Xue2} for any $a>0$. Therefore, the constraint \dref{E:sharpened3} is also stricter than \dref{E:Xue2}.
\end{enumerate}

A simple example demonstrating the above improvements is given in Appendix \ref{A:example}. The improvements 1) and 2) come from fixed composition code analysis \cite{Gallager} (or alternatively a time-sharing argument), while the key to improvement 3), which accounts for the structural change from \dref{E:Xue2} to \dref{E:sharpened3}, is a new argument for bounding $H(I_n|Y^n)$ instead of analyzing the successful decoding probability based only on $Y$ as done in \cite{Xue}. %We also point out that the $\sqrt a$ in \dref{E:Xue2} is sharpened to $\sqrt{\frac{a \ln 2}{2}}$ in \dref{E:sharpened3} because we use a stronger generalized blowing-up lemma \cite{Measure} than the one used in \cite{Xue}; the $\log |\Omega|$ in \dref{E:Xue2} is sharpened to  $\log (|\Omega|-1)$ in \dref{E:sharpened3} due to a more accurate characterization of the volume of a hamming ball in $\Omega^n$; and the range $[0, R_0]$ of $a$ in Proposition \ref{P:Xue} is modified to \dref{E:newconstraint_a} so that i) the function $H\left(\sqrt{\frac{a \ln 2}{2}}\right)$ is always meaningful, and ii) the search of the optimal $a$ may be more efficient.

\subsection{Via Bounding $H(Y^n|I_n)$}\label{SS:Newbound2}
 
Before presenting our second upper bound, we first define a parameter that will be used in stating the theorem.  %As will be clear from the proof of this bound, given a conditional probability distribution $p(\omega|x)$ this parameter  gives a single-letter characterization %(to the first order in the exponent) of the maximum possible ratio between two transition probabilities $p(\omega^n|x^n)$ and $p(\tilde{\omega}^n|x^n)$, i.e. $\max \frac{1}{n}\log \frac{p(\omega^n|x^n)}{p(\tilde{\omega}^n|x^n)}$, where the maximization is taken over all the $x^n$, $\omega^n$ and $\tilde{\omega}^{n}$ such that $(x^n,\omega^n)$ has empirical distribution $p(x)p(\omega|x)$ and the Hamming distance between $\omega^n$ and $\tilde{\omega}^{n}$ is no greater than $n\sqrt{\frac{a\ln 2}{2}}$, i.e., $d(\omega^n,\tilde{\omega}^{n})\leq n\sqrt{\frac{a\ln 2}{2}}$.

\begin{definition}\label{D:Delta}
Given a fixed channel transition probability $p(\omega|x)$, for any $p(x)$ and $d\geq 0$, $\Delta(p(x),d)$ is defined as
%\footnote{In the notation for $\Delta(p(x),a)$, we drop the dependence $p(\omega|x)$ since for us $p(\omega|x)$ is the channel fixed by the problem definition.}
\begin{align}
\Delta(p(x),d):=&\max_{\tilde p(\omega|x)}  H(\tilde p(\omega|x) |p(x))+D( \tilde p(\omega|x) || p(\omega|x) |p(x))- H(p(\omega|x) |p(x))\label{E:objective_def}  \\
\text{s.t.~~~~} & \frac{1}{2}\sum_{(x,\omega)}|p(x)\tilde p(\omega|x)-p(x)  p(\omega|x)|\leq  d.\label{E:constraint_def}
\end{align}
\end{definition}

In the above, we adopt the notation in \cite{Csiszarbook}. Specifically, $D( \tilde p(\omega|x) || p(\omega|x) |p(x))$ is the conditional relative entropy defined as
\begin{align}
D( \tilde p(\omega|x)||p(\omega|x)|p(x)):=  \sum_{(x,\omega) } p(x)\tilde p(\omega|x) \log \frac{\tilde p(\omega|x) }{p(\omega|x)} ,\label{E:conditionalrelativeentropy}
\end{align}
$H(\tilde p(\omega|x) |p(x))$ is the conditional entropy defined with respect to the joint distribution $p(x)\tilde p(\omega|x)$, i.e.,
\begin{align}
H(\tilde p(\omega|x) |p(x)):=  -\sum_{(x,\omega) } p(x)\tilde p(\omega|x) \log  \tilde p(\omega|x)  ,\label{E:conditionalentropy}
\end{align}
and $H(p(\omega|x) |p(x))$ is the conditional entropy similarly defined with respect to $p(x)p(\omega|x)$.

%As will be clear from the proof of our second bound, the parameter $\Delta(p(x),a)$ in fact gives a single-letter characterization (to the first order in the exponent) of the maximum possible ratio between two transition probabilities $p(\omega^n|x^n)$ and $p(\tilde{\omega}^n|x^n)$, i.e. $\max \frac{1}{n}\log \frac{p(\omega^n|x^n)}{p(\tilde{\omega}^n|x^n)}$, where the maximization is taken over all the $x^n$, $\omega^n$ and $\tilde{\omega}^{n}$ such that $(x^n,\omega^n)$ has composition $p(x)p(\omega|x)$ and $d(\omega^n,\tilde{\omega}^{n})\leq n\sqrt{a/2}$.

$\Delta(p(x),d)$ can be interpreted as follows: given  a random variable $X\sim p(x)$, assume we want to describe a related random variable $Y$. We use a code designed for the conditional distribution $p(w|x)$ for $Y$ given $X$, while $Y$ actually comes from a distribution $\tilde{p}(w|x)$. The distribution $\tilde{p}(w|x)$ cannot be too different than the assumed distribution in the sense  that the total variation distance between the two joint distributions $p(x)p(w|x)$ and $p(x)\tilde{p}(w|x)$ is bounded by $d$. $\Delta(p(x),d)$ captures the maximal inefficiency we would incur for having $Y$ come from a different distribution than the one assumed, i.e. it is  the maximal number of extra bits we would use when compared to the case where $Y$ comes from the assumed distribution.
 
%\begin{remark}\label{R:fullyconnect}
It can be easily seen that $\Delta(p(x),d)\geq 0$ for all $p(x)$ and $d\geq 0$, and  $\Delta(p(x),d)= 0$ when $d=0$. Moreover, for any fixed $p(x)$ and $d>0$, $\Delta(p(x),d)=\infty$ if and only if there exists some $x$ with $p(x)>0$, and some $\omega$ such that $p(\omega|x)=0$. Thus, a sufficient condition for $\Delta(p(x),d)<\infty$ for all $p(x)$ and $d> 0$ is that the channel
transition matrix is \emph{fully connected}, i.e., $p(\omega|x)>0,\forall (x,\omega)\in \Omega_{X}\times \Omega$. In this case, $\Delta(p(x),d)\to 0$ as $d\to 0$ for any $p(x)$.
%See Appendix \ref{A:delta}
%for the proofs of the properties of $\Delta(p(x),a)$.
%\end{remark}
\begin{example}
Suppose $p(\omega|x)$ corresponds  to a binary symmetric channel with crossover probability $p<1/2$. We derive $\Delta\left(p(x),d \right)$ according to Definition \ref{D:Delta} in Appendix \ref{A:derivebscdelta} and obtain  
\begin{align}\Delta\left(p(x),d \right)= \min\left\{ d,1-p   \right\} \log \frac{1-p}{p}. \end{align}
Interestingly, in this case $\Delta(p(x),d)$ has a simple expression that is independent of $p(x)$.
\end{example}

We are now ready to state our second new upper bound, which is proved by bounding $H(Y^n|I_n)$.

\begin{theorem}\label{T:novel}
For the symmetric primitive relay channel, if a rate $R$ is achievable, then there exists some $p(x)$ and
$a \in \left[0, \min\left\{R_0, H(Z|X)\right\} \right]$ such that
\begin{numcases}{}
R \leq I(X;Y,Z)   \label{E:ournewbound1}\\
 R  \leq I(X;Y)+R_0 - a \label{E:ournewbound2} \\
R \leq I(X;Y) + \Delta\left(p(x),\sqrt{\frac{a \ln 2}{2}}\right).  \label{E:ournewbound3}
\end{numcases}
\end{theorem}

Theorem \ref{T:novel} also implies the cut-set bound in Propositions \ref{P:cutset}. In particular, when the channels $X$-$Y$ and $X$-$Z$ have a fully connected
transition matrix, our new bound is \emph{strictly} tighter than the cut-set bound since $\Delta\left(p(x),d \right) \to 0$ as $d \to 0$ for any $p(x)$ in this case.

It should be pointed out that the bounds in Theorems \ref{T:sharpen} and \ref{T:novel} are proved based on essentially different arguments, and they do not include each other in general.
For instance, in the case when $X$-$Y$ and $X$-$Z$  are binary erasure channels (i.e. $\mbox{Pr}(Y=x|x)=1-p$, and $\mbox{Pr}(Y=\mbox{erasure}|x)=p$, $\forall x\in \{0,1\}$),
 $\Delta\left(p(x),d\right)= \infty$ for all $p(x)$ and $d> 0$, and thus our second bound reduces to the cut-set bound, but the first bound is still strictly tighter than the cut-set bound;
 whereas in the case when $X$-$Y$ and $X$-$Z$  are  binary symmetric channels,  our second bound is significantly  tighter  than both the cut-set bound and the first bound as we will show in the next section.

\section{Binary Symmetric Channel}\label{S:BSC}
As an application of the upper bounds stated in Sections \ref{S:existing}  and \ref{S:Newbounds}, we consider the case where the channel is binary symmetric, i.e., both the $X$-$Y$ and $X$-$Z$ links are binary symmetric channels with crossover probability $p$. The various upper bounds can be specialized to this case as follows (see Appendix \ref{A:BSCBounds} for derivations).
\begin{itemize}
  \item Cut-set bound (Prop. \ref{P:cutset}):
 \begin{align*}
  C(R_0) \leq  \min \left\{1+H(p*p)-2H(p), 1-H(p)+R_0   \right\},
  \end{align*}
where $p_1*p_2:=p_1(1-p_2)+p_2(1-p_1)$.
  \item Xue's bound (Prop. \ref{P:Xue}):
\begin{align*}
  C(R_0) \leq \max_{a\in [0,R_0] } \min \left\{1-H(p)+R_0-a, E^{-1}(H(\sqrt{a})+  \sqrt{a} )   \right\},
  \end{align*}
where $E^{-1}(\cdot)$ is the inverse function of $E(R)$.

  \item Our first bound (Thm. \ref{T:sharpen}):
  \begin{align*}
  C(R_0) \leq \max_{a\in [0, \min \{R_0, H(p),  \frac{1}{2\ln 2}\}] }  \min \Bigg \{1+H(p*p)-2H(p), \  &   1-H(p)+R_0-a ,    \\
&   1-H(p)+  H\left(\sqrt{\frac{a \ln 2}{2}}\right)-a         \Bigg\}.
  \end{align*}

  \item Our second bound (Thm. \ref{T:novel}):
\begin{align*}
  C(R_0) \leq \max_{a\in [0, \min \{R_0, H(p),\frac{2}{\ln 2}(1-p)^2\}] }  \min \Bigg \{1+H(p*p)-2H(p), \ &     1-H(p)+R_0-a ,    \\
&   1-H(p)+\sqrt{\frac{a \ln 2}{2}} \log \frac{1-p}{p} \Bigg\}.
  \end{align*}
\end{itemize}
Fig. \ref{F:upperbound} plots the above bounds for $p=0.2$ and $R_0\in [0.15,0.21]$. As can be seen, both of our bounds perform strictly better than the cut-set bound and Xue's bound, where the latter two are quite close to each other. Particularly, our second bound provides  considerable gain over the other three bounds.

\begin{figure}[hbt]
\centering
\includegraphics[width=0.7\textwidth]{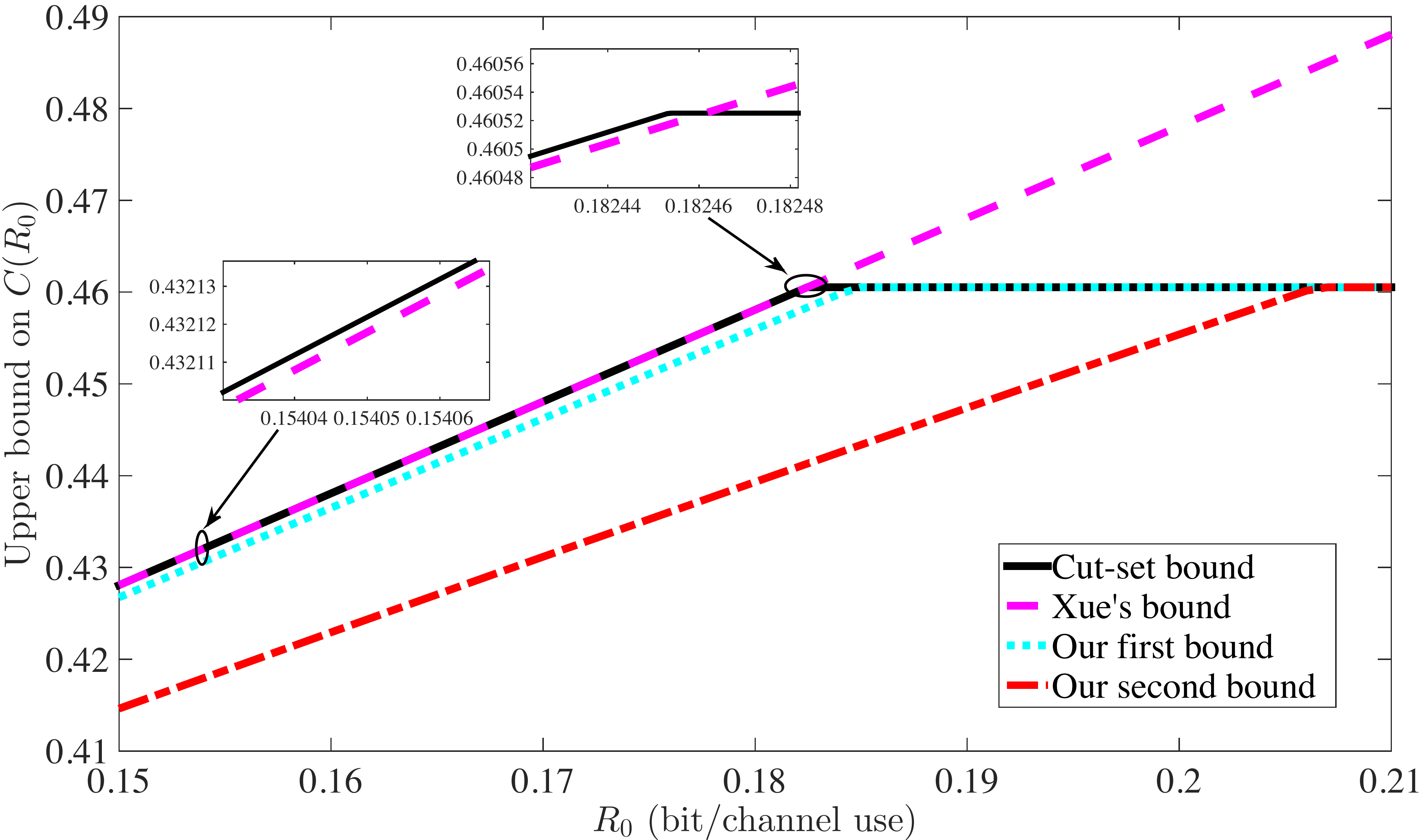}
\caption{Upper bounds on $C(R_0)$ for binary symmetric case with $p=0.2$.}
\label{F:upperbound}
\end{figure}

%\begin{figure*}[hbt]
%\centering
%\subfigure[Area $A$ in Fig.  \ref{F:upperbound}]{\includegraphics[width=0.485\textwidth]{upperboundC1}}
%\subfigure[Area $B$ in Fig.  \ref{F:upperbound}]{\includegraphics[width=0.485\textwidth]{upperboundC2}}
%\subfigure[Area $C$ in Fig.  \ref{F:upperbound}]{\includegraphics[width=0.485\textwidth]{upperboundC3}}
%\subfigure[Area $D$ in Fig.  \ref{F:upperbound}]{\includegraphics[width=0.485\textwidth]{upperboundC4}}
%\caption{Enlarged View of Areas $A$--$D$ in Fig.  \ref{F:upperbound}.}
%\label{F:enlarged}
%\end{figure*}

\subsection{Cover's Open Problem on the Critical $R_0$}

Now suppose we want to achieve the rate $C_{XYZ}$ for the relay channel. What is the minimum rate needed for the relay--destination link? This question was posed by Cover \cite{Cover'sopenproblem} and has been open for decades. Formally, we are interested in the critical value
$$R_0^*=\inf \{R_0: C(R_0)=C_{XYZ}\}.$$
The upper bounds on the capacity of the primitive relay channel presented in the previous sections  can be immediately used to develop lower bounds on $R_0^*$. Note that since Xue's bound is always dominated by our first bound, in the following we only compare the lower bounds on  $R_0^*$ implied by our two bounds with that implied by the cut-set bound.
\begin{itemize}
  \item Cut-set bound (Prop. \ref{P:cutset}):
  \begin{align*}
  R_0^* \geq H(p*p)-H(p).
  \end{align*}

  \item Our first bound (Thm. \ref{T:sharpen}):
  \begin{align*}
  R^*_0 \geq  \min_{H\left(\sqrt{\frac{a\ln 2}{2}}\right)-a \geq H(p*p)-H(p)} H(p*p)-H(p)+a.
  \end{align*}

  \item Our second bound (Thm. \ref{T:novel}):
  \begin{align*}
  R^*_0 \geq  H(p*p)-H(p)+ \frac{2}{\ln 2}\left(\frac{H(p*p)-H(p)}{\log \frac{1-p}{p}}\right)^2.
  \end{align*}
\end{itemize}
Fig. \ref{F:lowerbound} plots these lower bounds on $R_0^*$ versus the crossover probability $p$. We see again  that our second bound provides more gain over the cut-set bound than our first bound does.

\begin{figure}[hbt]
\centering
\includegraphics[width=0.7\textwidth]{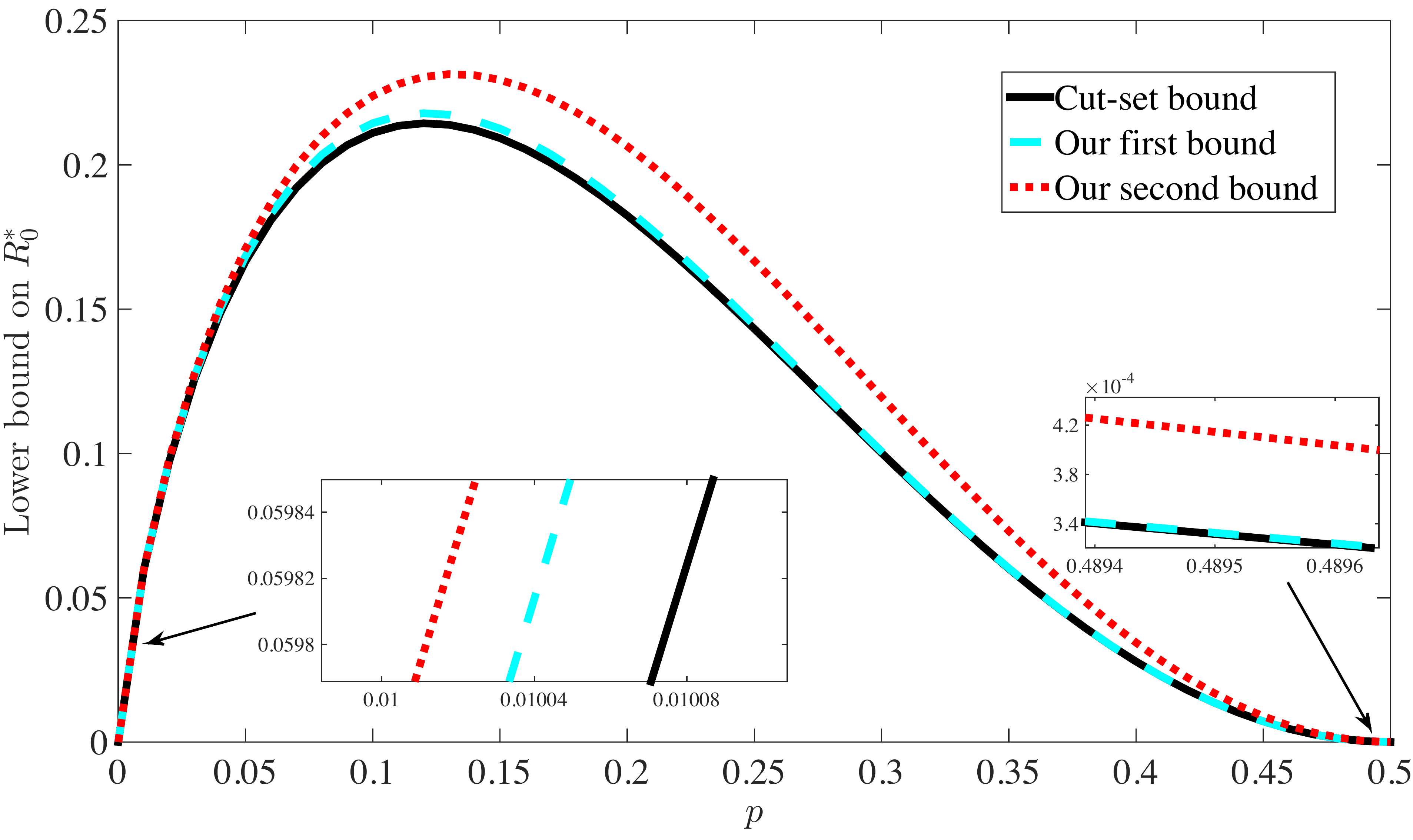}
\caption{Lower bounds on $R_0^*$ for binary symmetric case.}
\label{F:lowerbound}
\end{figure}

From Fig. \ref{F:lowerbound} we observe that  all these lower bounds on $R^*_0$ converge to $0$ as $p\to 0$ or $p\to 1/2$. On the other hand, to achieve $C_{XYZ}$, a natural way is to use a simple C-F scheme with only Slepian-Wolf binning, a.k.a. Hash-and-Forward (H-F) \cite{KimAllerton}, to faithfully transfer the relay's observation $Z^n$ to the destination so that the joint decoding based on $Z^n$ and $Y^n$ can be performed. This leads to an upper bound on $R^*_0$, namely $R_0^* \leq H(p*p)$, where $H(p*p)$ is the conditional entropy $H(Z|Y)$ induced by the uniform input distribution. Interestingly, this H-F upper bound also converges to 0 as $p\to 0$; but as $p\to 1/2$, it converges to 1 even though $C_{XYZ}$ is diminishing in this regime, which is in sharp contrast to the above lower bounds on $R^*_0$ that all converge to 0.

This leads to an interesting dichotomy: as $p\to 1/2$ while achievability requires a full bit of $R_0$ to support the diminishing $C_{XYZ}$ rate, the converse results allow for a diminishing $R_0$. Building on our second upper bound on the capacity of the primitive relay channel, in Section \ref{S:furtherimprovement} we prove the following improved lower bound on $R^*_0$, which deviates from 0 as $p\to 1/2$, thus suggesting that a positive $R_0$ is needed to achieve $C_{XYZ}$ even when $C_{XYZ}\to 0$. The proof of this result follows the argument for proving our second bound, however it also critically incorporates the fact that the rate of the codebook is approximately $C_{XYZ}$ in this case as well as the fact that the channel is binary symmetric, which allows us to do a combinatorial geometric analysis of the typical sets in Hamming space.

%we will employ a combinatorial geometric analysis\footnote{It should be pointed out that such a combinatorial geometric analysis relies on the geometry of Hamming spaces and the assumption that the rate of the relay channel is approximately $C_{XYZ}$. It may not be directly applicable for deriving a general upper bound on $C(R_0)$ like those in Theorems \ref{T:sharpen} and \ref{T:novel}.} in Hamming space to derive the following improved lower bound on $R_0^*$.
\begin{theorem}\label{T:further}
For the binary symmetric channel case,
$$R_0^*\geq   H(p*p)-H(p)+ \frac{2}{\ln 2}\left(\frac{H(p*p)-H(p)}{(1-2p)\log \frac{1-p}{p}}\right)^2.$$
\end{theorem}
\begin{proof}
See Section \ref{S:furtherimprovement}.
\end{proof}

Fig. \ref{F:lowerbound_im} shows this further improved lower bound on $R_0^*$ as well as the H-F upper bound. Clearly, this improved lower bound is tighter than all other lower bounds, and in particular, it converges to a strictly positive value, $0.1803$, as $p \to 1/2$ while all the other lower bounds converge to 0. This also shows that $R_0^*$ is discontinuous since  when $p=1/2$, the capacity of the relay channel is $0$, and therefore trivially $R_0^*=0$. We indeed believe that  $R_0^*\to 1$ as $p \to 1/2$ but proving this currently remains out of reach.

%Nevertheless, there is still a significant gap between the improved lower bound and the H-F upper bound, which remains to be narrowed. The exact characterization of $R_0^*$ is still unclear at this moment.

\begin{figure}[hbt]
\centering
\includegraphics[width=0.7\textwidth]{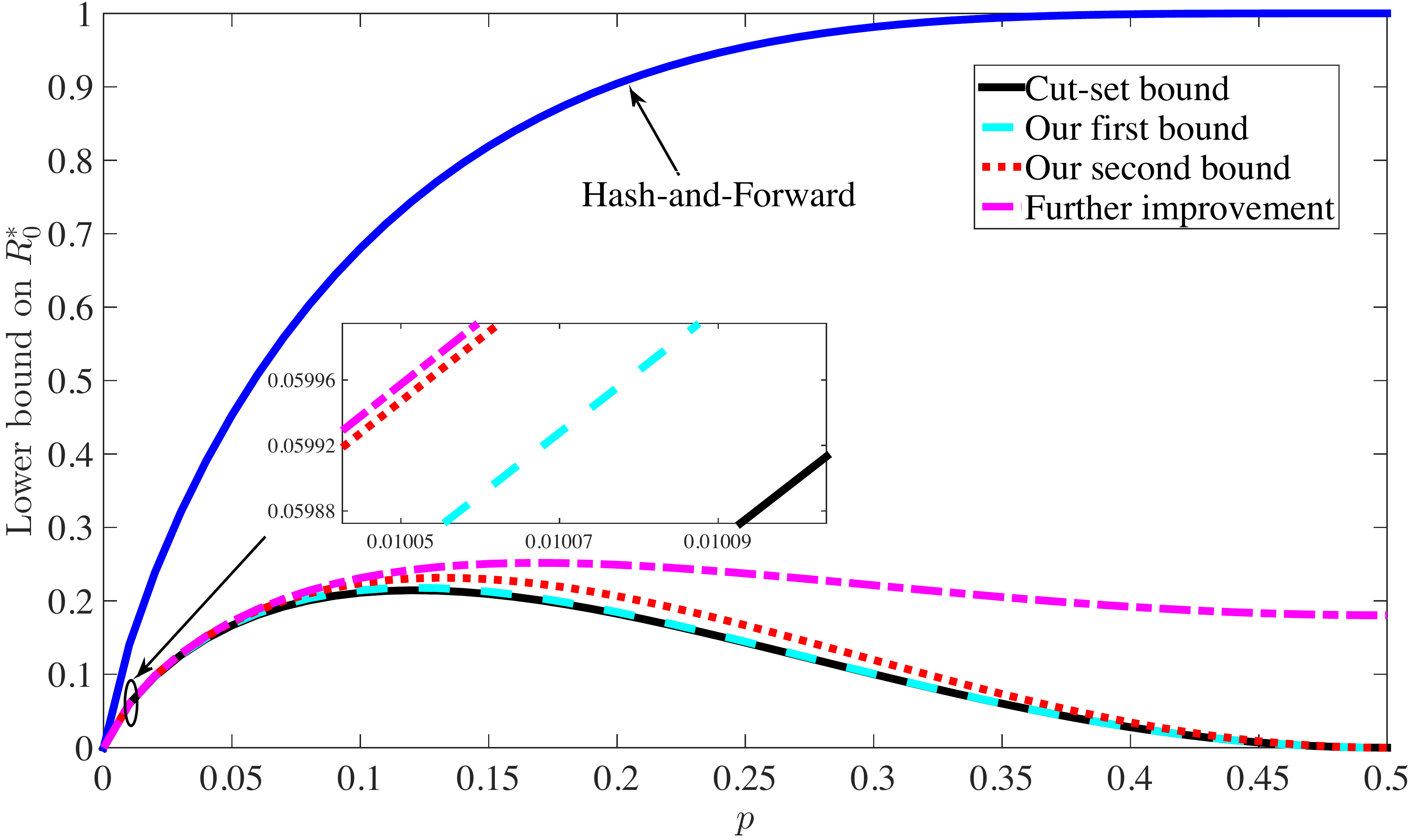}
\caption{Improved lower bound on $R_0^*$ for the binary symmetric case.}
\label{F:lowerbound_im}
\end{figure}

%As $p\to 0$, both the H-F upper bound and the lower bounds suggest a diminishing $R_0$; But as $p\to 1/2$, they differ from each other dramatically---in fact, they suggest two opposite extremes! This dichotomy brings up two immediate questions: Is H-F approximately optimal in the regime of $p\to 0$?  Do we need some positive $R_0$ or even a full bit of $R_0$ to support the diminishing $C_{XYZ}$ in the regime of $p\to 1/2$, or a diminishing $R_0$ is enough?

From Fig. \ref{F:lowerbound_im}, one can observe that in the other extreme, as $p \to 0$, upper and lower bounds on $R^*_0$ do indeed match and all approach $0$. One can indeed check that the speed at which they approach $0$ is also not too different. In particular, we can show that H-F is approximately optimal within a multiplicative factor of 2 in the regime where $p \to 0$. More precisely, letting $R_0^{\text{H-F}}=H(p*p)$ and $R_0^{\text{C-S}}=H(p*p)-H(p)$ denote the H-F bound and the cut-set bound on $R_0^*$ respectively, it can be shown (see Appendix \ref{A:optimality}) that
\begin{align}
\frac{R_0^{\text{H-F}}}{R_0^{\text{C-S}}}  \to 2 \text{ as } p \to 0.\label{E: HFopt}
\end{align}

%\end{proposition}
%\begin{proof}
%The proof is by routine calculus and is included in Appendix \ref{A:optimality}.
%\end{proof}

%Besides, in addition to using the techniques mentioned in Section \ref{SS:improvingcutset}, we will employ a combinatorial geometric analysis\footnote{It should be pointed out that such a combinatorial geometric analysis relies on the geometry of Hamming spaces and the assumption that the rate of the relay channel is approximately $C_{XYZ}$. It may not be directly applicable for deriving a general upper bound on $C(R_0)$ like those in Theorems \ref{T:sharpen} and \ref{T:novel}.}
% in Hamming space to derive the following improved lower bound on $R_0^*$. \begin{theorem}\label{T:further}
%For the binary symmetric channel case,
%$$R_0^*\geq   H(p*p)-H(p)+ \frac{2}{\ln 2}\left(\frac{H(p*p)-H(p)}{(1-2p)\log \frac{1-p}{p}}\right)^2.$$
%\end{theorem}
%\begin{proof}
%See Section \ref{S:furtherimprovement}.
%\end{proof}

\section{Proof of Theorem \ref{T:sharpen}}\label{S:sharpenproof}

In this section, we prove bounds \dref{E:sharpened1}--\dref{E:sharpened3}  sequentially with the focus on showing \dref{E:sharpened3}.

\subsection{Fixed Composition Code Argument}
We start with a fixed composition code argument \cite{Gallager}, which is useful for coupling bounds \dref{E:sharpened1}--\dref{E:sharpened3} together through the input distribution $p(x)$. For the purpose of showing Theorem \ref{T:sharpen}  such an argument can be replaced by using time sharing random variable \cite{Coverbook}, however the latter technique is not sufficient in deriving the bound in Theorem \ref{T:novel}; therefore for consistency,  this paper  employs the former argument for proving both Theorems \ref{T:sharpen} and \ref{T:novel}.

\begin{definition}
The composition $Q_{x^n}$ (or empirical probability distribution) of a sequence $x^n$ is the relative proportion of occurrences of each symbol of $\Omega_X$, i.e., $Q_{x^n}(a)=N(a|x^n)/n$ for all
$a \in \Omega_X$, where $N(a|x^n)$ is the number of times the symbol $a$ occurs in the sequence $x^n$.
\end{definition}

%Let $\mathcal{Q}_n$ denote the set of compositions with denominator $n$. If $Q \in \mathcal{Q}_n$, then the set of sequences of length $n$ and composition $Q$ is called the composition class of
%$Q$, denoted $T(Q)$, i.e.,
%\begin{align*}
%T(Q)=\{x^n\in \Omega_X^n :  Q_{x^n}=Q\}.
%\end{align*}

%\begin{definition}
%A code is said to be of fixed composition $Q$, denoted by $\mathcal{C}_{(n, R)}^{[Q]}$, if all the codewords have composition $Q$, i.e., if $x^n (m) \in T(Q), \forall m \in \{1,2,\ldots, 2^{nR}\}$.
%\end{definition}

\begin{definition}
A code for the primitive relay channel is said to be of fixed composition $Q$, denoted by $(\mathcal{C}_{(n,R)}^{[Q]}, f_n, g_n)$,
if all the codewords in $\mathcal{C}_{(n,R)}$ have the same composition $Q$.
%i.e., $x^n (m) \in T(Q), \forall m \in \{1,2,\ldots, 2^{nR}\}$.
\end{definition}

The following lemma says that if a rate $R$ is achievable by some sequence of codes, then there exists a sequence of fixed composition codes that can achieve essentially the same rate.

\begin{lemma}\label{L:fixed}
Suppose a rate $R$ is achievable over the primitive relay channel. Then for any $\tau>0$, there exists a sequence of fixed composition codes with rate $R_\tau:=R-\tau$
 \begin{align}
 \{(\mathcal{C}_{(n,R_\tau)}^{[Q_n]}, f_n, g_n)\}_{n=1}^{\infty} \label{E:fixedcompositioncode}
 \end{align}
such that the average probability of error $P_e^{(n)} \to 0$ as $n \to \infty$.
\end{lemma}
\begin{IEEEproof}
The proof relies on the property that there are only a polynomial
number of compositions and can be found in Appendix \ref{A:fixedproof}.
\end{IEEEproof}

\subsection{Proof of \dref{E:sharpened1}--\dref{E:sharpened2}}
To prove Theorem \ref{T:sharpen}, in the sequel we will use the reliable fixed composition codes in \dref{E:fixedcompositioncode}. A benefit of this is that now the various $n$-letter information quantities have single letter characterizations or bounds, as demonstrated in the following.
\begin{lemma}\label{L:entropy}
For the $n$-channel use code with fixed composition $Q_n$, we have
\begin{align*}
H(Y^n|X^n)=H(Z^n|X^n)&=nH(Y|X)=nH(Z|X)\\
H(Y^n,Z^n|X^n)&=nH(Y,Z|X),
\end{align*}
and
\begin{align*}
I(X^n;Y^n)=I(X^n;Z^n)&\leq nI(X;Y)=nI(X;Z)\\
I(X^n;Y^n,Z^n)&\leq nI(X;Y,Z)
\end{align*}
where $H(Y|X),H(Z|X)$ and $I(X;Y),I(X;Z)$ are calculated based on $Q_n(x)p(\omega|x)$, and $H(Y,Z|X)$ and $I(X;Y,Z)$ are calculated based on $Q_n(x)p(y|x)p(z|x)$.
\end{lemma}
\begin{proof}
See Appendix \ref{A:entropy}.
\end{proof}

Let the relay's transmission be denoted by $I_n=f_n(Z^n)$. With the above lemma, we have
\begin{align}
n R_\tau &= n(R-\tau)=H(M)\nonumber \\
&=I(M;Y^n,I_n)+H(M|Y^n,I_n)\nonumber \\
&\leq I (X^n;Y^n,I_n)+n \epsilon \label{E:Fano1} \\
&\leq I (X^n;Y^n,Z^n)+n \epsilon \nonumber \\
&\leq n(I(X;Y,Z)+ \epsilon) \nonumber
\end{align}
i.e.,
\begin{align}
R\leq I(X;Y,Z)+ \tau+ \epsilon \label{E:valid1}
\end{align}
for any $\tau, \epsilon >0$ and sufficiently large $n$, where \dref{E:Fano1} follows from Fano's inequality.

Moreover, for any $\tau, \epsilon >0$ and sufficiently large $n$, continuing with \dref{E:Fano1} we have
\begin{align}
n(R-\tau) &\leq I (X^n;Y^n,I_n)+n\epsilon \nonumber \\
&=I(X^n;Y^n)+I(X^n;I_n|Y^n)+ n\epsilon \nonumber \\
&=I(X^n;Y^n)+H(I_n|Y^n)-H(I_n|X^n)+ n\epsilon \label{E:newcontinue} \\
&\leq n(I(X;Y)+R_0 -a_n+ \epsilon)\label{E:newcontinue1} \end{align}
i.e.,
\begin{align}
 R \leq  I(X;Y)+R_0 -a_n+\tau+ \epsilon  \label{E:valid2}
\end{align}
where $a_n:=\frac{1}{n}H(I_n|X^n)$ is subject to the following constraint
\begin{align}
0\leq a_n\leq \min \left\{R_0, \frac{1}{n}H(Z^n|X^n) \right\}=  \min \{R_0,  H(Z|X) \}. \label{E:valid3}
\end{align}

\subsection{Proof of \dref{E:sharpened3}}

To prove \dref{E:sharpened3}, we continue with \dref{E:newcontinue}, and instead of upper bounding $H(I_n|Y^n)$ by $nR_0$ as in \dref{E:newcontinue1}, we use the following upper bound on $H(I_n|Y^n)$, whose proof will be given in Section \ref{SS:proofoflemma1}.

%again from \dref{E:Fano1}, we obtain that for any $\tau, \epsilon >0$ and sufficiently large $n$,
%\begin{align}
%n(R-\tau) &\leq I (X^n;Y^n,I_n)+n\epsilon \nonumber \\
%&=H(Y^n,I_n)-H(Y^n,I_n|X^n)+ n\epsilon \label{E:latercont} \\
%&=H(Y^n)+H(I_n|Y^n)- H(Y^n|X^n)-H(I_n|X^n)+ n\epsilon \nonumber \\
%&=H(I_n|Y^n)+I(X^n;Y^n)-na_n+ n\epsilon \nonumber\\
%&\leq H(I_n|Y^n)+n(I(X;Y) -a_n+ \epsilon) \label{E:toproceed}.
%\end{align}

%The following is the key lemma for proving Theorem \ref{T:sharpen}, which upper bounds the conditional entropy $H(I_n|Y^n)$ in \dref{E:toproceed} and whose proof will be given in Section \ref{SS:proofoflemma1}.
\begin{lemma}\label{L:keylemmafortheorem1}
For any fixed $n$,
$$H(I_n|Y^n)\leq nV\left(  \sqrt{\frac{a_n\ln 2}{2}}   \right),$$
with
\begin{numcases}{V(r):=}
\log |\Omega| & if $r>  \frac{|\Omega|-1}{|\Omega|}  $  \label{E:vdef1} \\
H(r)+r\log(|\Omega|-1) &  if $r \leq \frac{|\Omega|-1}{|\Omega|}  $.\label{E:vdef2}
\end{numcases}
where $H(r)$ is the binary entropy function defined as $H(r)=-r\log r -(1-r)\log (1-r)$.
\end{lemma}

Plugging the bound on $H(I_n|Y^n)$ in Lemma \ref{L:keylemmafortheorem1} into \dref{E:newcontinue}, we have for any $\tau, \epsilon >0$ and sufficiently large $n$,
\begin{align}
R\leq I(X;Y)+ V\left(  \sqrt{\frac{a_n\ln 2}{2}}   \right)-a_n+ \tau+\epsilon \label{E:valid4}
\end{align}

Combining \dref{E:valid1}, \dref{E:valid2}, \dref{E:valid4} and \dref{E:valid3}, we have that if a rate $R$ is achievable, then for any $\delta>0$ and sufficiently large $n$,\begin{numcases}{}
R \leq I(X;Y,Z)+\delta   \nonumber \\
 R \leq I(X;Y) + R_0- a_n +\delta \nonumber \\
R\leq I(X;Y)+ V\left(  \sqrt{\frac{a_n\ln 2}{2}}   \right)-a_n+\delta  \nonumber
 \end{numcases}
where
$$a_n \in [0,\min\{R_0 , H(Z|X)\}].$$ Since $\delta$ can be arbitrarily small, we arrive at the following proposition.

\begin{proposition}\label{P:Equivalence}
If a rate $R$ is achievable, then there exists some $p(x)$  and $a \in \left[0, \min\left\{R_0, H(Z|X)\right\} \right]$ such that
\begin{numcases}{}
R \leq I(X;Y,Z)  \label{PE:sharpened1}\\
 R \leq I(X;Y) + R_0- a \label{PE:sharpened2}\\
R\leq I(X;Y)+ V\left(  \sqrt{\frac{a\ln 2}{2}}   \right)-a \label{PE:sharpened3}
\end{numcases}
where
\begin{numcases}{V\left(  \sqrt{\frac{a\ln 2}{2}}   \right)=}
\log |\Omega| & if $a > \frac{2}{\ln 2} \left( \frac{|\Omega|-1}{|\Omega|} \right)^2$  \nonumber \\
H\left(  \sqrt{\frac{a\ln 2}{2}}   \right)+\sqrt{\frac{a\ln 2}{2}} \log(|\Omega|-1) &  if $a \leq \frac{2}{\ln 2} \left( \frac{|\Omega|-1}{|\Omega|} \right)^2$ \nonumber.
\end{numcases}
\end{proposition}

Now we show that Proposition \ref{P:Equivalence} is in fact equivalent to Theorem \ref{T:sharpen}.

  Theorem \ref{T:sharpen}  $\rightarrow$ Proposition \ref{P:Equivalence}:
Suppose Theorem \ref{T:sharpen} is true. Then, for any $R$ achievable, there exists some $$a \in \left[0, \min\left\{R_0, H(Z|X),  \frac{2}{\ln 2}  \left( \frac{|\Omega|-1}{|\Omega|} \right)^2\right\} \right]$$ satisfying \dref{E:sharpened1}--\dref{E:sharpened3}. For such $a \leq  \frac{2}{\ln 2}  \left( \frac{|\Omega|-1}{|\Omega|} \right)^2$,  \dref{PE:sharpened3} reduces to \dref{E:sharpened3} and thus Proposition \ref{P:Equivalence} is also true.

  Proposition \ref{P:Equivalence} $\rightarrow$ Theorem \ref{T:sharpen}:
Suppose Proposition \ref{P:Equivalence} is true. Then, for any $R$ achievable, there exists some $a \in \left[0, \min\left\{R_0, H(Z|X)\right\} \right]$ satisfying \dref{PE:sharpened1}--\dref{PE:sharpened3}. If such $a\leq  \frac{2}{\ln 2}  \left( \frac{|\Omega|-1}{|\Omega|} \right)^2$, then \dref{E:sharpened1}--\dref{E:sharpened3} hold with this $a$; otherwise, \dref{E:sharpened1}--\dref{E:sharpened3} hold with the choice of $a'=  \frac{2}{\ln 2}  \left( \frac{|\Omega|-1}{|\Omega|}\right)^2$. In either case, Theorem \ref{T:sharpen} is also true.

This establishes the equivalence between Proposition \ref{P:Equivalence} and Theorem \ref{T:sharpen}, and thus completes the proof of Theorem \ref{T:sharpen}.

\subsection{Proof of Lemma \ref{L:keylemmafortheorem1}}\label{SS:proofoflemma1}
To prove the inequality in Lemma \ref{L:keylemmafortheorem1} for any fixed $n$, we go to a higher dimensional, say $nB$ dimensional space,  to invoke the concepts of typical sets, and resort to a result on measure concentration, namely, the generalized blowing-up lemma.

Specifically, consider the $B$-length i.i.d. extensions of the random variables $X^n, Y^n, Z^n$ and  $I_n$,  i.e.,
\begin{align}
\{   (X^n(b),Y^n(b),Z^n(b),I_n(b)   )    \}_{b=1}^{B}, \label{E:iidextension}
\end{align}
where for any $b\in [1:B]$, $(X^n(b),Y^n(b),Z^n(b),I_n(b))$ has the same distribution as $(X^n, Y^n, Z^n,I_n)$.  For notational convenience,  in the sequel we write the $B$-length vector $[X^n(1),X^n(2),\ldots,X^n(B)]$ as $\mathbf X$ and similarly define $\mathbf Y, \mathbf Z$ and $\mathbf I$; note here we have
$\mathbf I=[f_n(Z^n(1)), f_n(Z^n(2)),\ldots,f_n(Z^n(B)) ]=:f(\mathbf Z)$.

The following lemma is critical for establishing Lemma \ref{L:keylemmafortheorem1}. We prove this lemma after we finish the proof of Lemma \ref{L:keylemmafortheorem1}. The proof is based on typicality arguments  combined with  the generalized blowing-up lemma.

\begin{lemma}\label{L:blowupsimilar}
Let $f^{-1}(\mathbf i):=\{\underline{\mathbf{\omega}}\in \Omega^{nB}: f(\underline{\mathbf{\omega}})=\mathbf i \}$ and $\Gamma_{nB (\sqrt{\frac{a_n\ln 2}{2}}+\delta) }(f^{-1}(\mathbf i))$ be its blown-up set defined as
\begin{align*}
\Gamma_{nB (\sqrt{\frac{a_n\ln 2}{2}}+\delta) }(f^{-1}(\mathbf i)):= \left\{\underline{\mathbf{\omega}}\in \Omega^{nB}: \exists \  \underline{\mathbf{\omega}}'\in f^{-1}(\mathbf i)\text{~~s.t.~~} d(\underline{\mathbf{\omega}},\underline{\mathbf{\omega}}')\leq nB \left(\sqrt{\frac{a_n\ln 2}{2}}+\delta\right)\right\},
\end{align*}
where $d(\underline{\mathbf{\omega}},\underline{\mathbf{\omega}}')$ denotes the Hamming distance between the two sequences $\underline{\mathbf{\omega}}$ and $\underline{\mathbf{\omega}}'$. Then for any $\delta>0$ and $B$ sufficiently large,
$$\mbox{Pr}(\mathbf{Y}\in \Gamma_{nB (\sqrt{\frac{a_n\ln 2}{2}}+\delta) }(f^{-1}(\mathbf I)) ) \geq 1-\delta.$$
\end{lemma}

With the above lemma, we now upper bound $H(\mathbf I |\mathbf Y)$. Let
$$E=\mathbb I (\mathbf{Y}\in \Gamma_{nB (\sqrt{\frac{a_n\ln 2}{2}}+\delta) }(f^{-1}(\mathbf I))) $$
where $\mathbb{I}(\cdot)$ is the
indicator function defined as
\begin{numcases}{\mathbb{I}(A)=}
1  \text{~~~if $A$ holds} \nonumber \\
0  \text{~~~otherwise. } \nonumber
\end{numcases}
We have
\begin{align}
H(\mathbf I |\mathbf Y)& \leq H(\mathbf I , E|\mathbf Y)\nonumber \\
&= H( E|\mathbf Y)+H(\mathbf I |\mathbf Y,E)\nonumber \\
&\leq H(\mathbf I |\mathbf Y,E)+1\nonumber \\
&=\mbox{Pr}(E=1) H(\mathbf I |\mathbf Y,E=1)+\mbox{Pr}(E=0) H(\mathbf I |\mathbf Y,E=0)+ 1\nonumber \\
&\leq H(\mathbf I |\mathbf Y,E=1) +\delta nBR_0 +1  .\label{E:plug}
\end{align}
To bound $H(\mathbf I |\mathbf Y,E=1) $, consider a Hamming ball centered at $\mathbf Y$ of radius $nB \left(\sqrt{\frac{a_n\ln 2}{2}}+\delta\right)$, which we denote as\footnote{The Hamming ball here should be distinguished from the notion of Hamming sphere that will be used later in Section \ref{S:furtherimprovement}. Specifically, a Hamming ball centered at $\mathbf{c}$ of radius $r$, denoted by $\mbox{Ball}(\mathbf{c},r)$, is defined as the set of points that are within Hamming distance $r$ of $\mathbf{c}$, whereas a corresponding Hamming sphere, denoted by $\mbox{Sphere}(\mathbf{c},r)$, is the set of points that are at a Hamming distance equal to $r$ from $\mathbf{c}$.
%In both cases, the center $\mathbf{c}$ can be omitted in the notation when it becomes irrelevant.
}
\begin{align*}
\text{Ball}\left( \mathbf Y,  nB \left(\sqrt{\frac{a_n\ln 2}{2}}+\delta\right) \right):= \left\{\underline{\mathbf{\omega}}: d(\underline{\mathbf{\omega}}, \mathbf Y)\leq   nB \left(\sqrt{\frac{a_n\ln 2}{2}}+\delta\right)   \right\}.
\end{align*}
 The condition $E=1$, i.e., $\mathbf{Y}\in \Gamma_{nB (\sqrt{\frac{a_n\ln 2}{2}}+\delta) }(f^{-1}(\mathbf I))$, ensures that there is at least one point $\underline{\omega} \in f^{-1}(\mathbf I)$ belonging to this ball, and therefore, given $E=1$ and $\mathbf Y$ the number of different possibilities for $\mathbf I$  is bounded by  $\left|\text{Ball}\left( \mathbf Y,  nB \left(\sqrt{\frac{a_n\ln 2}{2}}+\delta\right) \right)\right|$, the number of sequences in this Hamming ball, leading to the following upper bound on $H(\mathbf I |\mathbf Y,E=1) $,
\begin{align}
H(\mathbf I |\mathbf Y,E=1)& \leq \log \left|\text{Ball}\left( \mathbf Y,  nB \left(\sqrt{\frac{a_n\ln 2}{2}}+\delta\right) \right)\right|\nonumber \\
& = nB V\left( \sqrt{\frac{a_n\ln 2}{2}}+\delta\right) \label{E:volumne}\\
&\leq nB \left[V\left( \sqrt{\frac{a_n\ln 2}{2}}\right)+\delta_1\right] \label{E:Vcont}
\end{align}
for some $\delta_1\to 0$ as $\delta \to 0$, where the function $V(\cdot)$ is defined as in \dref{E:vdef1}--\dref{E:vdef2}, \dref{E:volumne} follows from the characterization of the volume of a Hamming ball (see Appendix \ref{A:ball} for details), and \dref{E:Vcont} follows from the continuity of the function $V(\cdot)$.
Plugging \dref{E:Vcont} into \dref{E:plug}, we have
\begin{align*}
H(\mathbf I |\mathbf Y) \leq  nB \left[V\left( \sqrt{\frac{a_n\ln 2}{2}}\right)+\delta_1\right]+\delta nBR_0 +1.
\end{align*}
Dividing $B$ at both sides of the above inequality and noting that
$$H(\mathbf I |\mathbf Y)=\sum_{b=1}^B H(I_n(b)|Y^n(b))=BH(I_n|Y^n),$$
we have
\begin{align}
H(I_n|Y^n) \leq  n\left( V\left( \sqrt{\frac{a_n\ln 2}{2}}\right)+\delta_1+\delta R_0+\frac{1}{nB}\right).\label{E:arbitrarysmall}
\end{align}
Since $\delta, \delta_1$ and $\frac{1}{nB}$ in \dref{E:arbitrarysmall} can all be made arbitrarily small by choosing $B$ sufficiently large, we obtain
\begin{align}
H(I_n|Y^n) \leq  nV\left( \sqrt{\frac{a_n\ln 2}{2}}\right).
\end{align}
This finishes the proof of Lemma \ref{L:keylemmafortheorem1}.

We are now in a position to prove Lemma \ref{L:blowupsimilar}. For this, we will apply the following generalized blowing-up lemma \cite[Lemma 12]{Measure}.

\begin{lemma}[Generalized Blowing-Up Lemma]\label{L:blowingup}
Let $U_1,U_2,\ldots, U_n$ be $n$ independent random variables taking values in a finite set $\mathcal{U}$. Then, for any $A \subseteq \mathcal{U} ^n$ with $\mbox{Pr}(U^n\in A)\geq 2^{-n a_n}$,
\begin{align*}
\mbox{Pr}(U^n\in\Gamma_{n(\sqrt{ \frac{a_n \ln 2}{2}}+r)} (A))\geq 1- e ^{-2nr^2}, \forall r>0.
\end{align*}
\end{lemma}

\begin{proof}[Proof of Lemma \ref{L:blowupsimilar}]
Consider any $(\mathbf{x},\mathbf{i} )\in \mathcal T_{\epsilon}^{(B)}(X^n, I_n)$, where $\mathcal T_{\epsilon}^{(B)}(X^n, I_n)$ denotes the $\epsilon$-jointly typical sets\footnote{This paper adopts the same definitions and notations for typical, jointly typical, and conditionally typical sets as those in \cite{ElGamalKim}.} with respect to $(X^n, I_n)$. From the property of jointly typical sequences (cf. \cite[Sec. 2.5]{ElGamalKim}), we have for some $\epsilon_1 \to 0$ as $\epsilon \to 0$,
\begin{align*}
p( \mathbf i  |\mathbf{x})& \geq 2^{-B( H(I_n|X^n)+\epsilon_1 )} \geq 2^{-nB( a_n+\epsilon_1)},
\end{align*}
i.e.,
$$\mbox{Pr}(\mathbf{Z} \in f^{-1}(\mathbf i) |\mathbf{x})\geq 2^{-nB( a_n+\epsilon_1 )}.$$
Note that due to the discrete memoryless property of the channel, given $\mathbf{x}$, $\mathbf Z$ is an $nB$-length sequence of independent random variables, and we then have, by applying Lemma \ref{L:blowingup}, that
\begin{align*}
\mbox{Pr}(\mathbf{Z}\in \Gamma_{nB (\sqrt{\frac{a_n\ln 2}{2}}+2\sqrt{\epsilon_1}) }(f^{-1}(\mathbf i) ) |\mathbf{x})
&=\mbox{Pr}(\mathbf{Z}\in \Gamma_{nB (\sqrt{\frac{(a_n+\epsilon_1)\ln 2}{2}} + [\sqrt{\frac{a_n\ln 2}{2}}+2\sqrt{\epsilon_1}-\sqrt{\frac{(a_n+\epsilon_1)\ln 2}{2}}])      }(f^{-1}(\mathbf i) ) |\mathbf{x})\\
&\geq \mbox{Pr}(\mathbf{Z}\in \Gamma_{nB (\sqrt{\frac{(a_n+\epsilon_1)\ln 2}{2}} + [\sqrt{\frac{a_n\ln 2}{2}}+2\sqrt{\epsilon_1}-\sqrt{\frac{a_n \ln 2}{2}}-\sqrt{\frac{\epsilon_1\ln 2}{2}}])      }(f^{-1}(\mathbf i) ) |\mathbf{x})\\
&\geq \mbox{Pr}(\mathbf{Z}\in \Gamma_{nB (\sqrt{\frac{(a_n+\epsilon_1)\ln 2}{2}} + \sqrt{\epsilon_1} )      }(f^{-1}(\mathbf i) ) |\mathbf{x})\\
&\geq 1- e^{-2nB\epsilon_1 }\\
&\geq 1-\sqrt{\epsilon_1}
\end{align*}
for sufficiently large $B$. Noting that $\mathbf{Y}$ and $\mathbf{Z}$ are identically distributed given $\mathbf{X}$, we obtain
$$\mbox{Pr}(\mathbf{Y}\in \Gamma_{nB (\sqrt{\frac{a_n\ln 2}{2}}+2\sqrt{\epsilon_1}) }(f^{-1}(\mathbf i) ) |\mathbf{x})\geq 1-\sqrt{\epsilon_1},  $$
and thus,
\begin{align}
\mbox{Pr}(\mathbf{Y}\in \Gamma_{nB (\sqrt{\frac{a_n\ln 2}{2}}+2\sqrt{\epsilon_1}) }(f^{-1}(\mathbf I) ))&= \sum_{(\mathbf{x},\mathbf{i})}\mbox{Pr}(\mathbf{Y}\in \Gamma_{nB (\sqrt{\frac{a_n\ln 2}{2}}+2\sqrt{\epsilon_1}) }(f^{-1}(\mathbf i) )| \mathbf{x},\mathbf{i}   )p(\mathbf{x},\mathbf{i})\nonumber \\
&=  \sum_{(\mathbf{x},\mathbf{i})}\mbox{Pr}(\mathbf{Y}\in \Gamma_{nB (\sqrt{\frac{a_n\ln 2}{2}}+2\sqrt{\epsilon_1}) }(f^{-1}(\mathbf i) )| \mathbf{x}   )p(\mathbf{x},\mathbf{i})\label{E:markov}\\
&\geq  \sum_{(\mathbf{x},\mathbf{i})\in \mathcal T_{\epsilon}^{(B)}(X^n, I_n)}\mbox{Pr}(\mathbf{Y}\in \Gamma_{nB (\sqrt{\frac{a_n\ln 2}{2}}+2\sqrt{\epsilon_1}) }(f^{-1}(\mathbf i) )| \mathbf{x}   )p(\mathbf{x},\mathbf{i})\nonumber\\
&\geq (1-\sqrt{\epsilon_1})  \sum_{(\mathbf{x},\mathbf{i})\in \mathcal T_{\epsilon}^{(B)}(X^n, I_n)}p(\mathbf{x},\mathbf{i}) \nonumber\\
&\geq (1-\sqrt{\epsilon_1}) ^2  \label{E:approach1}\\
&\geq 1-2\sqrt{\epsilon_1} \nonumber
\end{align}
for sufficiently large $B$, where \dref{E:markov} follows due to the Markov chain: $\mathbf Y  \leftrightarrow \mathbf X \leftrightarrow \mathbf Z \leftrightarrow\mathbf I$, and \dref{E:approach1} follows since $\text{Pr}(\mathcal T_{\epsilon}^{(B)}(X^n, I_n))\to 1$ as $B \to \infty$. Finally, choosing $\delta$ to be $2\sqrt{\epsilon_1}$ concludes the proof of Lemma \ref{L:blowupsimilar}.
\end{proof}

\section{Proof of Theorem \ref{T:novel}}\label{S:novelproof}

The bounds \dref{E:ournewbound1}--\dref{E:ournewbound2} are the same as \dref{E:sharpened1}--\dref{E:sharpened2}, which have been proved in Section \ref{S:sharpenproof}. To show \dref{E:ournewbound3}, still consider the reliable fixed composition codes in \dref{E:fixedcompositioncode}. Then we have the following lemma, which upper bounds the conditional entropy $H(Y^n|I_n)$ and whose proof is given in Section \ref{SS:conditionalentropy2}.

\begin{lemma}\label{L:conditionalentropy2}
For any $n$-channel use code with fixed composition $Q_n$,
\begin{align}
H(Y^n|I_n)\leq H(X^n|I_n) -H(X^n|Z^n) + n H(Y|X) +n \Delta\left(Q_n,\sqrt{\frac{a_n \ln 2}{2}}\right), \label{E:conditionalentropy2}
\end{align}
where $H(Y |X)$ is calculated based on $Q_n(x)p(\omega|x)$, $a_n=\frac{1}{n}H(I_n|X^n)$, and $\Delta(\cdot,\cdot)$ is as defined in \dref{E:objective_def}--\dref{E:constraint_def}.

\end{lemma}

With this lemma, we then have
\begin{align}
n(R-\tau)  &\leq I(X^n;Y^n,I_n) + n\epsilon   \nonumber\\
&= I(X^n;I_n) +I(X^n;Y^n|I_n) + n\epsilon   \nonumber\\
&= H(X^n)- H(X^n|I_n) + H(Y^n|I_n)- H(Y^n|X^n)  + n\epsilon   \nonumber\\
&\leq H(X^n)- H(X^n|I_n) + \left[H(X^n|I_n) -H(X^n|Z^n) + nH(Y|X) +n\Delta\left(Q_n,\sqrt{\frac{a_n \ln 2}{2}}\right)\right]\nonumber\\
&~~~~- H(Y^n|X^n)  + n\epsilon   \nonumber\\
      &= I(X^n;Z^n) +  n\Delta\left(Q_n,\sqrt{\frac{a_n \ln 2}{2}}\right) +n\epsilon  \label{E:cancel} \\
 &\leq n\left[ I(X;Y) +  \Delta\left(Q_n,\sqrt{\frac{a_n \ln 2}{2}}\right)+\epsilon\right] \label{E:symmetryyz}
\end{align}
for any $\tau, \epsilon>0$ and $n$ sufficiently large, where in \dref{E:cancel} we have used the fact that $H(Y^n|X^n)=n H(Y|X)$ (cf. Lemma \ref{L:entropy}), and \dref{E:symmetryyz} follows from the symmetry between $Y^n$ and $Z^n$ and Lemma \ref{L:entropy} again. This proves the bound \dref{E:ournewbound3} and hence Theorem \ref{T:novel}.

\subsection{Proof of Lemma \ref{L:conditionalentropy2}}\label{SS:conditionalentropy2}
The remaining step then is to show the entropy inequality \dref{E:conditionalentropy2} in Lemma \ref{L:conditionalentropy2}. For this, again we look at the $B$-length i.i.d. sequence of $(X^n, Y^n, Z^n, I_n)$, i.e.,
$$(\mathbf X, \mathbf Y, \mathbf Z, \mathbf I):=\{ (X^n(b), Y^n(b), Z^n(b), I_n(b)) \}_{b=1}^B.$$
The following lemma is crucial for proving inequality \dref{E:conditionalentropy2}, and its own proof will be given in the next subsection.
\begin{lemma}\label{L:keytoentropy2}
For any $\delta>0$ and $B$ sufficiently large, there exists a set $\mathcal S(Y^n, I_n)$ of $(\mathbf y, \mathbf i)$ pairs such that
\begin{align*}
\mbox{Pr}( (\mathbf Y, \mathbf I) \in  \mathcal S(Y^n, I_n)) \geq 1-\delta,
\end{align*}
and for any $(\mathbf y, \mathbf i) \in \mathcal S(Y^n, I_n)$,
\begin{align*}
 p(\mathbf{y}|\mathbf{i})\geq 2^{- B(H(X^n|I_n) -H(X^n|Z^n) + nH(Y|X) +n\Delta\left(Q_n,\sqrt{\frac{a_n \ln 2}{2}}\right)+ \delta)} .
\end{align*}
\end{lemma}

We will now use Lemma~\ref{L:keytoentropy2} to prove Lemma~\ref{L:conditionalentropy2}. Letting $E=\mathbb I ((\mathbf Y, \mathbf I) \in  \mathcal S(Y^n, I_n))$, we have for any $\delta>0$ and $B$ sufficiently large,
\begin{align}
H(\mathbf Y |\mathbf I)& \leq H(\mathbf Y , E|\mathbf I)\nonumber \\
&= H( E|\mathbf I)+H(\mathbf Y|\mathbf I,E)\nonumber \\
&\leq H(\mathbf Y |\mathbf I,E)+1\nonumber \\
&=\mbox{Pr}(E=1) H(\mathbf Y |\mathbf I,E=1)+\mbox{Pr}(E=0) H(\mathbf Y |\mathbf I,E=0)+ 1\nonumber \\
&\leq H(\mathbf Y |\mathbf I,E=1) +\delta nB\log |\Omega| +1 \nonumber  \\
&= -\sum_{(\mathbf y, \mathbf i) \in  \mathcal S(Y^n, I_n)} p(\mathbf y, \mathbf i |E=1) \log  p(\mathbf y| \mathbf i , E=1) +\delta nB\log |\Omega| +1\nonumber \\
& \leq  -\sum_{(\mathbf y, \mathbf i) \in  \mathcal S(Y^n, I_n)} p(\mathbf y, \mathbf i |E=1) \log  p(\mathbf y| \mathbf i ) +\delta nB\log |\Omega| +1\nonumber\\
&\leq B\left[H(X^n|I_n) -H(X^n|Z^n) + nH(Y|X) +n\Delta\left(Q_n,\sqrt{\frac{a_n \ln 2}{2}}\right)+ \delta\right]+\delta nB\log |\Omega| +1\label{E:div}.
\end{align}
Dividing $B$ at both sides of \dref{E:div} and noticing that $H(\mathbf Y |\mathbf I)=BH(Y^n|I_n)$, we have
\begin{align*}
H(Y^n|I_n)\leq H(X^n|I_n) -H(X^n|Z^n) + nH(Y|X) +n\Delta\left(Q_n,\sqrt{\frac{a_n \ln 2}{2}}\right)+ \delta +\delta n\log |\Omega| +\frac{1}{B}.
\end{align*}
Since both $\delta$ and $\frac{1}{B}$ in the above inequality can be made arbitrarily small by choosing $B$ sufficiently large, Lemma \ref{L:conditionalentropy2} is thus proved.

\subsection{Proof of Lemma \ref{L:keytoentropy2}}

Let $\mathcal S(Y^n, I_n)$ be defined as
\begin{align}
\mathcal S(Y^n, I_n):=\{ (\mathbf y, \mathbf i):  \mathbf{y}\in \Gamma_{nB (\sqrt{\frac{a_n\ln 2}{2}}+\epsilon) }( \mathcal T_{\epsilon}^{(B)}(Z^n|\mathbf i)  ) \}.\label{E:exactdef}
\end{align}
We first show that for any $\epsilon>0$ and $B$ sufficiently large,
\begin{align}
\mbox{Pr}( (\mathbf Y, \mathbf I) \in \mathcal S(Y^n, I_n) ) \geq 1-\epsilon.\label{E:chooseepsilon}
\end{align}
For this, consider any $(\mathbf x, \mathbf i)\in \mathcal T_{\tilde \epsilon}^{(B)}(X^n,I_n)$, $\tilde \epsilon>0$. By the joint typicality lemma (cf. \cite[Sec. 2.5.1]{ElGamalKim}), we have
\begin{align*}
\mbox{Pr}(\mathbf Z \in \mathcal T_{\tilde \epsilon}^{(B)}(Z^n| \mathbf x, \mathbf i)  |  \mathbf x) &  \geq 2^{-B(I( Z^n;I_n  |X^n)  +\tilde \epsilon_1)}   \\
&= 2^{-B(H(I_n  |X^n)  +\tilde \epsilon_1)}   \\
&\geq  2^{-nB(a_n  +\tilde \epsilon_1)},
\end{align*}
where $\tilde \epsilon_1\to 0$ as $\tilde \epsilon \to 0$ and $B\to \infty$. Since $T_{\tilde \epsilon}^{(B)}(Z^n| \mathbf x, \mathbf i)  \subseteq T_{\tilde \epsilon}^{(B)}(Z^n| \mathbf i)$, we further have
\begin{align*}
\mbox{Pr}(\mathbf Z \in \mathcal T_{\tilde \epsilon}^{(B)}(Z^n| \mathbf i)  |  \mathbf x) &   \geq  2^{-nB(a_n  +\tilde \epsilon_1)}.
\end{align*}
Then, by applying Lemma \ref{L:blowingup} along the same lines as the proof of Lemma \ref{L:blowupsimilar}, we can obtain
\begin{align}
\mbox{Pr}(\mathbf{Y}\in \Gamma_{nB (\sqrt{\frac{a_n\ln 2}{2}}+2\sqrt{\tilde\epsilon_1}) }(\mathcal T_{\tilde \epsilon}^{(B)}(Z^n| \mathbf i) ))\geq 1-2\sqrt{\tilde \epsilon_1} \nonumber
\end{align}
for sufficiently large $B$. Choosing $\epsilon$ to be $\max\{2\sqrt{\tilde \epsilon_1},\tilde \epsilon\}$ then proves \dref{E:chooseepsilon}.

Consider any $(\mathbf y, \mathbf i) \in \mathcal S(Y^n, I_n)$. By the definition of $\mathcal S(Y^n, I_n)$, we can find one
$\mathbf{z}\in \mathcal T_{ \epsilon}^{(B)}(Z^n| \mathbf i)$ such that
\begin{align}\label{E:hammingbound}
d (\mathbf{y},\mathbf{z})\leq nB \left(\sqrt{\frac{a_n\ln 2}{2}}+\epsilon\right).
\end{align}
Then,
\begin{align}
p(\mathbf{y}|\mathbf{i})&=\sum_{\mathbf x} p(\mathbf{y}|\mathbf{x})p(\mathbf{x}|\mathbf{i})\nonumber \\
&\geq \sum_{\mathbf x \in \mathcal T_{\epsilon}^{(B)}(X^n| \mathbf z, \mathbf i)} p(\mathbf{y}|\mathbf{x})p(\mathbf{x}|\mathbf{i})\label{E:throughout} \\
&\geq 2^{-B(H(X^n|I_n)+\epsilon_1)} \sum_{\mathbf x \in \mathcal T_{\epsilon}^{(B)}(X^n| \mathbf z, \mathbf i)} p(\mathbf{y}|\mathbf{x})\label{E:proseq} \\
&\geq 2^{-B(H(X^n|I_n)+\epsilon_1)} \big| \mathcal T_{\epsilon}^{(B)}(X^n| \mathbf z, \mathbf i) \big| \min_{\mathbf x \in \mathcal T_{\epsilon}^{(B)}(X^n| \mathbf z, \mathbf i)} p(\mathbf{y}|\mathbf{x}) \label{E:throughout1} \\
&\geq 2^{-B(H(X^n|I_n)+\epsilon_1)} 2^{B(H(X^n|Z^n)-\epsilon_2)} \min_{\mathbf x \in \mathcal T_{\epsilon}^{(B)}(X^n| \mathbf z, \mathbf i)} p(\mathbf{y}|\mathbf{x}), \label{E:probtobecont}
\end{align}
for some $\epsilon_1,\epsilon_2 \to 0$ as $\epsilon \to 0$ and $B \to \infty$, where the $\mathbf z$ throughout
\dref{E:throughout}--\dref{E:probtobecont} is the one belonging to $\mathcal  T_{ \epsilon}^{(B)}(Z^n| \mathbf i)$ and satisfying \dref{E:hammingbound},  and \dref{E:proseq} and \dref{E:probtobecont} follow from the properties of jointly typical sequences (cf. \cite[Sec. 2.5]{ElGamalKim}).

We  now lower bound $p(\mathbf{y}|\mathbf{x})$ for any $\mathbf x \in \mathcal T_{\epsilon}^{(B)}(X^n| \mathbf z, \mathbf i)$. Since $\mathbf x \in \mathcal T_{\epsilon}^{(B)}(X^n| \mathbf z, \mathbf i)$, we have $(\mathbf x,\mathbf z)  \in \mathcal T_{\epsilon}^{(B)}(X^n, Z^n)$, i.e., $(\mathbf x,\mathbf z) $ are jointly typical with respect to the $n$-letter random variables $(X^n, Z^n)$. Due to the fixed composition code assumption and the discrete memoryless property of the channel, this can be shown (see Appendix \ref{A:jointlytypical}) to further imply that $(\mathbf x,\mathbf z) $ are also jointly typical with respect to the single-letter random variables $(X, Z)$,   i.e.,
\begin{align}
|P_{(\mathbf{x},\mathbf{z})}(x ,\omega ) -Q_n(x)  p(\omega|x)| \leq \epsilon_3 Q_n(x)  p(\omega|x), \label{E:l1bound1}
\end{align}
for some $\epsilon_3 \to 0$ as $\epsilon \to 0$,
where $P_{(\mathbf{x},\mathbf{z})}(x ,\omega ) $ denotes the joint empirical distribution of $(\mathbf{x},\mathbf{z})$ with respect to  $(X, Z)$, defined as
\begin{align*}
P_{(\mathbf{x},\mathbf{z})}(x,\omega)&= \frac{1}{nB}N(x,\omega|\mathbf x, \mathbf z)
\end{align*}
where $N(x,\omega|\mathbf x, \mathbf z)$ denotes the number of times the  symbols $(x,\omega)$ occur in the sequences $(\mathbf x, \mathbf z)$.
On the other hand, we show in Appendix \ref{A:totalvariance} that the bound \dref{E:hammingbound} on the Hamming distance between $\mathbf{y}$ and $\mathbf{z}$ can translate to a bound on the total variation distance between the two empirical distributions $P_{(\mathbf{x},\mathbf{y})}(x ,\omega )$ and $P_{(\mathbf{x},\mathbf{z})}(x ,\omega )$ for any $\mathbf x$, namely,
\begin{align}
\sum_{(x,\omega)}|P_{(\mathbf{x},\mathbf{y})}(x ,\omega )-P_{(\mathbf{x},\mathbf{z})}(x ,\omega ) |\leq \frac{2}{nB} d (\mathbf{y},\mathbf{z})\leq  2 \left(\sqrt{\frac{a_n\ln 2}{2}}+\epsilon\right).\label{E:l1bound2}
\end{align}
Combining \dref{E:l1bound1} and \dref{E:l1bound2}, we have for some $\epsilon_4 \to 0$ as $\epsilon \to 0$,
\begin{align}
\sum_{(x,\omega)}|P_{(\mathbf{x},\mathbf{y})}(x ,\omega )-Q_n(x)  p(\omega|x) |\leq  2 \sqrt{\frac{a_n\ln 2}{2}}+\epsilon_4,
\end{align}
or equivalently expressed as
\begin{align}
\frac{1}{2}\sum_{(x,\omega)}|Q_n(x)P_{\mathbf{y}|\mathbf{x}}(\omega|x)-Q_n(x)  p(\omega|x)|\leq    \sqrt{\frac{a_n\ln 2}{2}}+\frac{\epsilon_4}{2},\label{E:l1bound3}
\end{align}
where we have used the fact that the empirical distribution $P_{\mathbf x}(x)=Q_n(x)$ due to the fixed composition code assumption, and $P_{\mathbf{y}|\mathbf{x}}(\omega|x)$ is the conditional empirical distribution satisfying
\begin{align*}
 N(x,\omega|\mathbf x, \mathbf y)= N(x |\mathbf x ) P_{\mathbf{y}|\mathbf{x}}(\omega|x) .
\end{align*}

To bound $p(\mathbf{y}|\mathbf{x})$, we have
\begin{align}
-\frac{1}{nB}\log  p(\mathbf{y}|\mathbf{x})  &=-\frac{1}{nB} \sum_{i=1}^{nB} \log p(y_i|x_i)\nonumber \\
&= -\sum_{(x,\omega) } P_{(\mathbf{x},\mathbf{y})}(x,w) \log p(\omega|x)\nonumber \\
&= \sum_{(x,\omega) }[-P_{(\mathbf{x},\mathbf{y})}(x,w) \log p(\omega|x)+P_{(\mathbf{x},\mathbf{y})}(x,w) \log P_{\mathbf{y}|\mathbf{x}}(w|x)-P_{(\mathbf{x},\mathbf{y})}(x,w) \log P_{\mathbf{y}|\mathbf{x}}(w|x)      ]\nonumber \\
&= -\sum_{(x,\omega) } P_{(\mathbf{x},\mathbf{y})}(x,w) \log P_{\mathbf{y}|\mathbf{x}}(w|x)
+ \sum_{(x,\omega) }  P_{(\mathbf{x},\mathbf{y})}(x,w) \log \frac{P_{\mathbf{y}|\mathbf{x}}(w|x)}{p(\omega|x)}  \nonumber \\
&= H(P_{\mathbf{y}|\mathbf{x}}(\omega|x) |P_{\mathbf x}(x))+D( P_{\mathbf{y}|\mathbf{x}}(\omega|x)  || p(\omega|x) |P_{\mathbf x}(x))\nonumber    \\
&= H(P_{\mathbf{y}|\mathbf{x}}(\omega|x) |Q_n(x))+D( P_{\mathbf{y}|\mathbf{x}}(\omega|x)  || p(\omega|x) |Q_n(x)),  \label{E:deltacom1}
\end{align}
where $P_{\mathbf{y}|\mathbf{x}}(\omega|x)$ satisfies the constraint \dref{E:l1bound3}. For any  $p(x)$ and $d\geq 0$, define $\Delta\left(p(x),d \right)$  as follows:
\begin{align}
\Delta\left(p(x),d\right):=&\max_{\tilde p(\omega|x)}  H(\tilde p(\omega|x) |p(x))+D( \tilde p(\omega|x) || p(\omega|x) |p(x))- H(p(\omega|x) |p(x))\label{E:deltacom2}\\
\text{s.t.~~~~~~~~~~} & \frac{1}{2}\sum_{(x,\omega)}|p(x)\tilde p(\omega|x)-p(x)  p(\omega|x)|\leq   d\label{E:deltacom3}
\end{align}

Comparing \dref{E:deltacom1} and \dref{E:l1bound3} to \dref{E:deltacom2} and \dref{E:deltacom3}, we have
\begin{align}
-\frac{1}{nB}\log  p(\mathbf{y}|\mathbf{x})&\leq \Delta\left(Q_n,\sqrt{\frac{a_n \ln 2}{2}}+\frac{\epsilon_4}{2}\right)+ H(Y |X)\nonumber\\
&\leq \Delta\left(Q_n,\sqrt{\frac{a_n \ln 2}{2}}\right)+ H(Y |X)+\epsilon_5\label{eqpygivenx}
\end{align}
for some $\epsilon_5\to 0$ as $\epsilon\to 0$, where $H(Y |X)$ is calculated based on $Q_n(x)p(\omega|x)$, and the last inequality follows since  $\Delta(p(x),d)$ is continuous in $d$ for $d>0$.
This combined with  \dref{E:probtobecont} yields that for any $(\mathbf y, \mathbf i) \in \mathcal S(Y^n, I_n)$,
\begin{align*}
p(\mathbf y|\mathbf i)&\geq  2^{-B(H(X^n|I_n)+\epsilon_1)} 2^{B(H(X^n|Z^n)-\epsilon_2)}2^{-nB(\Delta\left(Q_n,\sqrt{\frac{a_n \ln 2}{2}}\right)+H(Y|X)+\epsilon_5)}\\
&\geq 2^{- B(H(X^n|I_n) -H(X^n|Z^n) + nH(Y|X) +n\Delta\left(Q_n,\sqrt{\frac{a_n \ln 2}{2}}\right)+ \epsilon_6)}
\end{align*}
for some $\epsilon_6 \to 0$ as $\epsilon \to 0$ and $B\to \infty$. Finally, choosing $\delta=\max\{\epsilon, \epsilon_6\}$, we have
$$\mbox{Pr}(\mathbf Y, \mathbf I) \in \mathcal S(Y^n, I_n) \geq 1-\delta,$$
and for any $(\mathbf y, \mathbf i) \in \mathcal S(Y^n, I_n)$,
\begin{align*}
p(\mathbf y|\mathbf i)\geq  2^{- B(H(X^n|I_n) -H(X^n|Z^n) + nH(Y|X) +n\Delta\left(Q_n,\sqrt{\frac{a_n \ln 2}{2}}\right)+ \delta)},
\end{align*}
which concludes the proof of Lemma \ref{L:keytoentropy2}.

\section{Proof of Theorem \ref{T:further}}\label{S:furtherimprovement}

The main idea for proving Theorem \ref{T:further} follows that for Theorem \ref{T:novel}. In order to highlight the difference, we first look at the the parameter $\Delta(p(x),d)$ that plays an important role in the bound in Theorem \ref{T:novel} more closely. In Section~\ref{SS:Newbound2}, we have indicated that $\Delta(p(x),d)$ can be interpreted as the maximal number of extra bits we would need to compress $Y$ given $X$, when $Y$ comes from a conditional distribution $\tilde{p}(w|x)$ instead of the assumed distribution $p(w|x)$ and the total variation distance between the two joint distributions is bounded by $d$. An alternative role that emerges for this quantity in the context of  the proof of Theorem~\ref{T:novel} is the following.

Consider a pair $(\mathbf x , \mathbf z)$ of $nB$-length sequences that are jointly typical with respect to $p(x)p(\omega|x)$. We have
\begin{align}
p( \mathbf z| \mathbf x)\doteq 2^{-nBH(p(\omega|x)|p(x))}.\label{E:discussioncomp1}
\end{align}
Let $\mathbf y$ be a sequence taking values in the same alphabet as $\mathbf z$ and bounded in its Hamming distance to $\mathbf z$ by $nBd$. Theorem~\ref{T:novel} is based on obtaining a lower bound on the conditional probability of the sequence $\mathbf y$ given $\mathbf x$ under $p(x)p(\omega|x)$. In particular, in \eqref{eqpygivenx}, we show that
\begin{align}
p(\mathbf y|\mathbf x) \stackrel {.}{\geq}   2^{-nB[ H(p(\omega|x)|p(x))+ \Delta(p(x),d) ]}. \label{E:discussioncomp2}
\end{align}
Comparing \dref{E:discussioncomp1} and \dref{E:discussioncomp2}, we can see that $\Delta(p(x),d)$ characterizes the maximum possible exponential decrease from $p( \mathbf z| \mathbf x)$ to $p( \mathbf y| \mathbf x)$ where $(\mathbf x , \mathbf z)$ is jointly typical with respect to $p(x)p(\omega|x)$ and the Hamming distance between $\mathbf y$ and $\mathbf z$ is bounded by $nBd$.

For the binary symmetric channel, i.e. when the conditional distribution $p(\omega|x)$ corresponds  to a binary symmetric channel with crossover probability $p<1/2$, we show in Appendix \ref{A:derivebscdelta} that we have the following explicit expression
\begin{align}
\Delta(p(x),d) &=    \min\left\{H(p)+ d \log \frac{1-p}{p} , -\log p \right \}  -H(p)\nonumber\\
&= \min\left \{ d ,   1-p  \right \} \log \frac{1-p}{p}.\label{eq:binary}
\end{align}
We next provide an alternative way to obtain this expression by resorting to the above interpretation of  $\Delta(p(x),d)$. Note that when $p(\omega|x)$ corresponds to a binary symmetric channel with crossover portability $p<1/2$, for a $(\mathbf x , \mathbf z)$ pair that is jointly typical with respect to $p(x)p(\omega|x)$, we have $d(\mathbf x , \mathbf z)\leq nB(p+\epsilon)$ and
\begin{align}p( \mathbf z| \mathbf x)\doteq 2^{-nBH(p)}.\label{E:comparebsc1}\end{align}
If $\mathbf y$ satisfies $d (\mathbf{y},\mathbf{z})\leq nB d$, then by the triangle inequality we have
\begin{align*}
d (\mathbf{x},\mathbf{y})&\leq d (\mathbf{x},\mathbf{z})+d (\mathbf{y},\mathbf{z})\\
&\leq nB (p+ d+2\epsilon)
\end{align*}
and therefore,
\begin{align*}
p( \mathbf y| \mathbf x)  &\stackrel {.}{\geq}  p^{nB (p+ d) }(1-p)^{ nB-nB (p+ d) } \\
&= 2^{-nB[H(p)+d \log \frac{1-p}{p}    ]  } .
\end{align*}
Since we also trivially have $p( \mathbf y| \mathbf x) \geq p^{nB}=2^{nB\log p}$, it follows that
\begin{align}
p( \mathbf y| \mathbf x)  &\stackrel {.}{\geq}   2^{-nB \min\left \{H(p)+ d\log \frac{1-p}{p}  , -\log p  \right\} } \label{E:comparebsc2}.
\end{align}
Comparing \dref{E:comparebsc1} with \dref{E:comparebsc2}, we have the maximum possible exponential decrease from $p( \mathbf z| \mathbf x)$ to $p( \mathbf y| \mathbf x)$ given by \eqref{eq:binary}.

The above discussion reveals that the proof of Theorem \ref{T:novel}  inherently uses the triangle inequality to obtain a worst case bound, equal to $nB(d+p)$, on the distance between $\mathbf x$ and $\mathbf y$. The new ingredient in the proof of Theorem \ref{T:further} is a more precise analysis on the distance between $\mathbf y$  and $\mathbf x$ by building on the fact that the capacity of the primitive relay channel is equal to the broadcast bound in the context of Cover's open problem. Specifically, we show that most of the typical $\mathbf x$'s are within a distance $nB(d * p)$ from $\mathbf y$ in this case. The detailed proof of Theorem \ref{T:further} is as follows, where we only emphasize the difference from that of Theorem \ref{T:novel}.

We start by observing that Theorem \ref{T:further} follows from the following proposition as a corollary.
\begin{proposition}\label{P:BSCequi}
In the binary symmetric channel case, for any $\tau>0$, if a rate $R=C_{XYZ}-\tau$ is achievable, then there exists some $a\geq 0$ such that
\begin{numcases}{}
C_{XYZ}-\tau \leq C_{XY}+R_0-a  \label{E:bscp1}  \\
C_{XYZ}-\tau  \leq C_{XY} +\Delta'\left(\sqrt{\frac{a\ln2}{2}} \right)+\mu   \label{E:bscp2}
\end{numcases}
where
\begin{align}
\Delta'\left(d \right):=d(1-2p)\log \frac{1-p}{p}\label{E:newdelta}
\end{align}
and $\mu \to 0$ as $\tau \to 0$.
\end{proposition}

Specifically, we have by \dref{E:bscp1},
 \begin{align}
 R_0&\geq C_{XYZ}-C_{XY} +a-\tau \nonumber \\
 &= 1+H(p*p)-2H(p)-(1-H(p))+a-\tau \nonumber \\
 & = H(p*p)-H(p) +a-\tau,\label{E:R_01}
 \end{align}
 where we use the fact that $C_{XYZ}=1+H(p*p)-2H(p)$ and $C_{XY}=1-H(p)$,
 and by \dref{E:bscp2},
  \begin{align*}
\Delta'\left(\sqrt{\frac{a\ln2}{2}} \right)&=\sqrt{\frac{a\ln 2}{2}}(1-2p)\log \frac{1-p}{p}\\
 &\geq   C_{XYZ}-C_{XY} -\tau-\mu \\
&= H(p*p)-H(p) -\tau-\mu
 \end{align*}
so that
   \begin{align}
a\geq \frac{2}{\ln 2}\left(\frac{H(p*p)-H(p)}{(1-2p)\log \frac{1-p}{p}}\right)^2-\mu_1\label{E:R_02}
 \end{align}
for some $\mu_1 \to 0$ as $\tau\to 0$. Combining \dref{E:R_01} and \dref{E:R_02}, we have
\begin{align*}
R_0\geq  H(p*p)-H(p) +\frac{2}{\ln 2}\left(\frac{H(p*p)-H(p)}{(1-2p)\log \frac{1-p}{p}}\right)^2 -\tau -\mu_1,
\end{align*}
and by the definition of $R_0^*$,
\begin{align*}
R^*_0 &\geq \lim_{\tau\to 0} H(p*p)-H(p) + \frac{2}{\ln 2}\left(\frac{H(p*p)-H(p)}{(1-2p)\log \frac{1-p}{p}}\right)^2 -\tau -\mu_1\\
&= H(p*p)-H(p) + \frac{2}{\ln 2}\left(\frac{H(p*p)-H(p)}{(1-2p)\log \frac{1-p}{p}}\right)^2
\end{align*}
which is Theorem \ref{T:further}.

We now show Proposition \ref{P:BSCequi}, whose proof builds on the technique developed to prove Theorem \ref{T:novel} but doesn't require fixed composition code analysis. To show \dref{E:bscp1}, along the lines of the proof of  \dref{E:ournewbound2}, we have for any achievable rate $R=C_{XYZ}-\tau$, $\tau >0$,
 \begin{align}
n(C_{XYZ}-\tau) &\leq I(X^n;Y^n)+nR_0-na_n +n\epsilon \nonumber \\
 & \leq n (C_{XY}+R_0-a_n +\epsilon ) ,\label{E:memoryless}
 \end{align}
 i.e.,
 \begin{align}
C_{XYZ}-\tau  \leq C_{XY}+R_0-a_n+\epsilon,  \label{E:bscequi1}
 \end{align}
 where \dref{E:memoryless} follows from the memoryless property of the channel and $\epsilon\to 0$ as  $n\to \infty$.
To show \dref{E:bscp2}, we need the following lemma, whose proof is given in Section \ref{SS:bsckeylemma}.
\begin{lemma}\label{L:conditionalentropybsc}
In the binary symmetric channel case, for any $n$-channel use code with rate $R=C_{XYZ}-\tau$ and $P_{e}^{(n)}\to 0$,
\begin{align}
H(Y^n|I_n)\leq H(X^n|I_n) -H(X^n|Z^n) + n H(Y|X) +n\Delta'\left(\sqrt{\frac{a_n\ln2}{2}} \right)+ n \mu, \label{E:conditionalentropybsc}
\end{align}
where $\mu$ can be made arbitrarily small by choosing $n$ sufficiently large and $\tau$ sufficiently small, $H(Y|X)=H(p)$, $a_n=\frac{1}{n}H(I_n|X^n)$, and $\Delta'(\cdot)$ is as defined in \dref{E:newdelta}.
\end{lemma}

With the above lemma, following the lines that lead to \dref{E:cancel} from \dref{E:conditionalentropy2}, we can show that for any achievable rate $R=C_{XYZ}-\tau$,  $\tau >0$,
 \begin{align*}
n(C_{XYZ}-\tau) &\leq I(X^n;Y^n)+n\Delta'\left(\sqrt{\frac{a_n\ln2}{2}} \right)+n   \mu +n\epsilon\\
 & \leq n \left[C_{XY}+\Delta'\left(\sqrt{\frac{a_n\ln2}{2}} \right)+  \mu + \epsilon \right],
 \end{align*}
 i.e.,
 \begin{align}
C_{XYZ}-\tau   \leq C_{XY}+\Delta'\left(\sqrt{\frac{a_n\ln2}{2}} \right)+ \mu+\epsilon  , \label{E:bscequi2}
 \end{align}
 where  $\mu\to0$ as $n\to \infty$ and $\tau \to 0$, and $\epsilon  \to 0$ as $n \to \infty$.  Combining \dref{E:bscequi1}  and \dref{E:bscequi2}  proves Proposition \ref{P:BSCequi}.

\subsection{Proof of Lemma \ref{L:conditionalentropybsc}}\label{SS:bsckeylemma}
To show the entropy inequality \dref{E:conditionalentropybsc} in Lemma \ref{L:conditionalentropybsc}, we again look at the $B$-length i.i.d. sequence of the $n$-letter random variables $(X^n, Y^n, Z^n, I_n)$ that are induced by the $n$-channel use reliable code of rate $C_{XYZ}-\tau$, denoted by
$$(\mathbf X, \mathbf Y, \mathbf Z, \mathbf I):=\{ (X^n(b), Y^n(b), Z^n(b), I_n(b)) \}_{b=1}^B.$$

\begin{lemma}\label{L:keytoentropybsc}
Given any $\delta>0$, for $\tau$ sufficiently small and $n,B$ sufficiently large, there exists a set $\mathcal S(Y^n, I_n)$ of $(\mathbf y, \mathbf i)$ pairs such that
\begin{align*}
\mbox{Pr}( (\mathbf Y, \mathbf I) \in  \mathcal S(Y^n, I_n)) \geq 1-\delta,
\end{align*}
and for any $(\mathbf y, \mathbf i) \in \mathcal S(Y^n, I_n)$,
\begin{align}
 p(\mathbf{y}|\mathbf{i})\geq 2^{- B(H(X^n|I_n) -H(X^n|Z^n) + nH(Y|X) +n\Delta'\left(\sqrt{\frac{a_n\ln2}{2}} \right)+ n\delta)} \label{E:BSCSHARPEN}.
\end{align}
\end{lemma}

With Lemma \ref{L:keytoentropybsc}, along the same lines as in the proof of Lemma \ref{L:conditionalentropy2}, we can show that
\begin{align*}
H(Y^n|I_n)\leq H(X^n|I_n) -H(X^n|Z^n) + nH(Y|X) +n\Delta'\left(\sqrt{\frac{a_n\ln2}{2}} \right)+ n\delta +\delta n\log |\Omega| +\frac{1}{B},
\end{align*}
where $\delta$ can be made arbitrarily small by choosing $\tau$ sufficiently small and $n,B$ sufficiently large. This proves the entropy inequality \dref{E:conditionalentropybsc}.

We are now in a position to show Lemma \ref{L:keytoentropybsc}.

\begin{proof}[Proof of Lemma \ref{L:keytoentropybsc}]
The only difference of Lemma \ref{L:keytoentropybsc} from Lemma \ref{L:keytoentropy2} is that here the lower bound on $p(\mathbf{y}|\mathbf{i}), \forall (\mathbf y, \mathbf i)\in  \mathcal S(Y^n, I_n)$ is sharpened to that of \dref{E:BSCSHARPEN}.
In particular, assume  $\mathcal S(Y^n, I_n)$ is defined exactly as in \dref{E:exactdef}. Then, for any specific $(\mathbf y_0, \mathbf i_0) \in \mathcal S(Y^n, I_n)$, we can find one
$\mathbf{z}_0\in \mathcal T_{ \epsilon}^{(B)}(Z^n| \mathbf i)$ such that
\begin{align}
d (\mathbf{y}_0,\mathbf{z}_0)\leq nB (\sqrt{\frac{a_n\ln 2}{2}}+\epsilon).
\end{align}
The key to the aforementioned sharpening is a tighter upper bound on the distance between $\mathbf{y}_0$ and the $\mathbf x$'s typical with $\mathbf z_0$, as stated in the following lemma. The proof of this lemma is based on a combinatorial geometric argument, and is deferred until we finish the proof of Lemma \ref{L:keytoentropybsc}.

\begin{lemma}\label{L:geoanalysis}
Consider any $\mathbf y_0$ such that $d(\mathbf y_0, \mathbf z_0)=nBd_0$ for some $\mathbf{z}_0\in \mathcal T_{ \epsilon}^{(B)}(Z^n)$. There exists some $\epsilon'\to 0$ as $\epsilon \to 0$ such that
\begin{align*}
\mbox{Pr}( d(\mathbf X, \mathbf y_0) \leq  nB ( d_0*p+\epsilon')   |\mathbf z_0)\geq 1- \upsilon
\end{align*}
where $\upsilon$ can be made arbitrarily small by choosing $n,B$ sufficiently large and $\tau$ sufficiently small.
\end{lemma}

Due to the above lemma, we have for some $\epsilon'\to 0$ as $\epsilon \to 0$,
\begin{align*}
  &\mbox{Pr}(\mathbf X \in
  \text{Ball} ( \mathbf y_0, nB  (   (\sqrt{\frac{a_n\ln 2}{2}}+\epsilon)*p    +\epsilon'))\bigcap    \mathcal T_{\epsilon}^{(B)}(X^n| \mathbf z_0, \mathbf i_0)              |\mathbf z_0)\\
  \geq \ &1-  \mbox{Pr}(\mathbf X \notin
  \text{Ball} ( \mathbf y_0, nB  (   (\sqrt{\frac{a_n\ln 2}{2}}+\epsilon)*p    +\epsilon'))|\mathbf z_0)-\mbox{Pr}(\mathbf X \notin
  \mathcal T_{\epsilon}^{(B)}(X^n| \mathbf z_0, \mathbf i_0)   |\mathbf z_0 )\\
   \geq \ &   1-\upsilon -\upsilon \\
=\ &   1-2\upsilon ,
  \end{align*}
 where we have used the fact that  $\mbox{Pr}(\mathbf X \notin
  \mathcal T_{\epsilon}^{(B)}(X^n| \mathbf z_0, \mathbf i_0)   |\mathbf z_0 )\to 0$ as $B\to \infty$.
  Since for any $\mathbf x \in   \mathcal T_{\epsilon}^{(B)}(X^n| \mathbf z_0, \mathbf i_0)$, $p(\mathbf x|\mathbf z_0)\leq 2^{-B(H(X^n|Z^n)-\epsilon_1)}$ for some $\epsilon_1\to 0$ as $\epsilon \to 0$, we have
\begin{align*}
|\text{Ball} ( \mathbf y_0, nB  (   (\sqrt{\frac{a_n\ln 2}{2}}+\epsilon)*p    +\epsilon')) \bigcap    \mathcal T_{\epsilon}^{(B)}(X^n| \mathbf z_0, \mathbf i_0)   |\geq (1-2\upsilon )2^{B(H(X^n|Z^n)-\epsilon_1)}\geq  2^{B(H(X^n|Z^n)-\epsilon_1-\upsilon_1)},
\end{align*}
where $\upsilon _1\to0$ as $\tau\to 0$ and $n,B\to \infty$.
Therefore,
\begin{align}
p(\mathbf y_0|\mathbf i_0)&\geq \sum_{ \mathbf x \in \text{Ball} ( \mathbf y_0, nB  (   (\sqrt{\frac{a_n\ln 2}{2}}+\epsilon)*p    +\epsilon')) \bigcap    \mathcal T_{\epsilon}^{(B)}(X^n| \mathbf z_0, \mathbf i_0)  }p(\mathbf y_0|\mathbf x)p(\mathbf x|\mathbf i_0)\nonumber\\
&\geq 2^{B(H(X^n|Z^n)-\epsilon_1-\upsilon_1)}   2^{-B(H(X^n|I_n)+\epsilon_2)}  \min_{ \mathbf x \in \text{Ball} ( \mathbf y_0, nB  (   (\sqrt{\frac{a_n\ln 2}{2}}+\epsilon)*p    +\epsilon'))  }p(\mathbf y_0|\mathbf x)    \label{E:distance}
\end{align}
where $\epsilon_2 \to 0$ as $\epsilon \to 0$.
For any $\mathbf x$ with $d(\mathbf x, \mathbf y_0) \leq  nB (  (\sqrt{\frac{a_n\ln 2}{2}}+\epsilon)*p+\epsilon')=:nB(\sqrt{\frac{a_n\ln 2}{2}}*p +\epsilon_3 ) $, we have
\begin{align}
p(\mathbf y_0|\mathbf x) &= (1-p)^{nB-d(\mathbf x, \mathbf y_0)}p^{d(\mathbf x, \mathbf y_0)}\nonumber \\
&= (1-p)^{nB(1-p)}p^{nBp}\cdot   \frac{(1-p)^{nB-d(\mathbf x, \mathbf y_0)}p^{d(\mathbf x, \mathbf y_0)}}{(1-p)^{nB(1-p)}p^{nBp}} \nonumber \\
&= 2^{-nBH(p)}\cdot  \left( \frac{p}{1-p}\right)^{d(\mathbf x, \mathbf y_0)-nBp} \nonumber \\
&\geq 2^{-nBH(p)}\cdot  \left( \frac{p}{1-p}\right)^{nB(\sqrt{\frac{a_n\ln 2}{2}}*p +\epsilon_3 )-nBp}  \nonumber \\
&= 2^{-nBH(p)}\cdot  \left( \frac{p}{1-p}\right)^{nB( \sqrt{\frac{a_n\ln 2}{2}}(1-2p) +\epsilon_3 )}  \nonumber \\
&= 2^{-nB\left(H(p) + ( \sqrt{\frac{a_n\ln 2}{2}}(1-2p) +\epsilon_3 )\log \frac{1-p}{p}\right)} \nonumber \\
&= 2^{-nB\left(H(p) +  \Delta'\left(\sqrt{\frac{a_n\ln2}{2}} \right)+\epsilon_4  \right)} \label{E:bscplug}
\end{align}
where $\epsilon_3,\epsilon_4\to 0$ as $\epsilon \to 0$. Plugging \dref{E:bscplug} into \dref{E:distance}, we obtain that
\begin{align*}
p(\mathbf y_0|\mathbf i_0)\geq  2^{-B(H(X^n|I_n)-H(X^n|Z^n)+nH(p)+n\Delta'\left(\sqrt{\frac{a_n\ln2}{2}} \right) +\upsilon_1+\epsilon_1+\epsilon_2 +n\epsilon_4 )},
\end{align*}
which proves the lemma.
\end{proof}

\subsection{Proof of Lemma \ref{L:geoanalysis}}
Consider a specific $(\mathbf{y}_0,\mathbf{z}_0)$ pair where $\mathbf{z}_0\in \mathcal T_{ \epsilon}^{(B)}(Z^n)$ and $d(\mathbf{y}_0,\mathbf{z}_0)=nBd_0$. Let $$q_d :=\mbox{Pr}( d(\mathbf X, \mathbf y_0) \geq  nB d   |\mathbf z_0).$$ To show Lemma \ref{L:geoanalysis}, we show that there exists some $\epsilon' \to 0$ as $\epsilon\to 0$ such that for $d=d_0*p+\epsilon'$, $q_d\leq \upsilon$, for some $\upsilon$ satisfying $\lim_{\tau \to 0} \lim_{n\to \infty} \lim_{B\to \infty} \upsilon =0$.

First, using the properties of jointly typical sequences, we can show (see Appendix \ref{A:bscproperty}) that there exists some $\epsilon_0\to 0$ as $\epsilon\to 0$ such that for any $\delta > 0$ and $B$ sufficiently large,
\begin{align}
\mbox{Pr}(p(\mathbf{Y}|\mathbf{z}_0)\leq 2^{-B(H(Y^n|Z^n)-\epsilon_0)}, d(\mathbf{Y},\mathbf{z}_0) \in [ nB(p*p- \epsilon_0),  nB(p*p+ \epsilon_0)]    |\mathbf{z}_0) \geq 1- \delta. \label{E:qd2}
\end{align}

Then consider the following inequalities:
\begin{align}
\mbox{Pr}(  d(\mathbf Y, \mathbf y_0) \geq  nB (d*p -\epsilon_0 )   |\mathbf z_0)
&=\sum_{\mathbf x} \mbox{Pr}(  d(\mathbf Y, \mathbf y_0) \geq  nB (d*p -\epsilon_0 )|\mathbf x)p(\mathbf x|\mathbf z_0)\nonumber \\
&\geq \sum_{\mathbf x: d(\mathbf x, \mathbf y_0)\geq  nBd } \mbox{Pr}(  d(\mathbf Y, \mathbf y_0) \geq  nB (d*p -\epsilon_0 )|\mathbf x)p(\mathbf x|\mathbf z_0)\nonumber \\
&\geq q_d \cdot \min_{\mathbf x: d(\mathbf x, \mathbf y_0)\geq nBd } \mbox{Pr}(  d(\mathbf Y, \mathbf y_0) \geq  nB (d*p -\epsilon_0 )|\mathbf x).\label{E:plugd1d}
\end{align}
Without loss of generality, consider a specific pair $(\mathbf x, \mathbf y_0)$ as shown in Fig. \ref{F:xy_0},
where $\mathbf y_0=\mathbf 0$ and $d(\mathbf x, \mathbf y_0)=nBd_1\geq nBd$.
By the law of large numbers, we have for any $\delta>0$ and $B$ sufficiently large,
\begin{align}
1-\delta &\leq \mbox{Pr}( \frac{1}{nB} N(0,1|\mathbf x,\mathbf Y)\geq   \frac{1}{nB} N(0|\mathbf x) p  -\frac{\epsilon_0}{2},
  \frac{1}{nB} N(1,1|\mathbf x,\mathbf Y)\geq   \frac{1}{nB} N(1|\mathbf x) (1-p)-\frac{\epsilon_0}{2} |\mathbf x )\nonumber \\
&\leq \mbox{Pr}(  N(1|\mathbf Y)\geq  nB(d_1* p -\epsilon_0)|\mathbf x)\nonumber \\
&\leq \mbox{Pr}(  N(1|\mathbf Y)\geq  nB(d* p -\epsilon_0)|\mathbf x)\label{E:d1d}\\
&= \mbox{Pr}(  d(\mathbf Y, \mathbf y_0)  \geq  nB(d* p -\epsilon_0)|\mathbf x),\nonumber
\end{align}
where \dref{E:d1d} follows since the event $N(1|\mathbf Y)\geq  nB(d_1* p -\epsilon_0)$ implies $N(1|\mathbf Y)\geq  nB(d* p -\epsilon_0)$ due to the relation $d_1\geq d$. Plugging this into \dref{E:plugd1d}, we obtain
\begin{align}
\mbox{Pr}(  d(\mathbf Y, \mathbf y_0) \geq  nB (d*p -\epsilon_0 )   |\mathbf z_0)
\geq q_d  (1-\delta).\label{E:qd1}
\end{align}

\begin{figure}[hbt]
\centering
\includegraphics[width=0.4\textwidth]{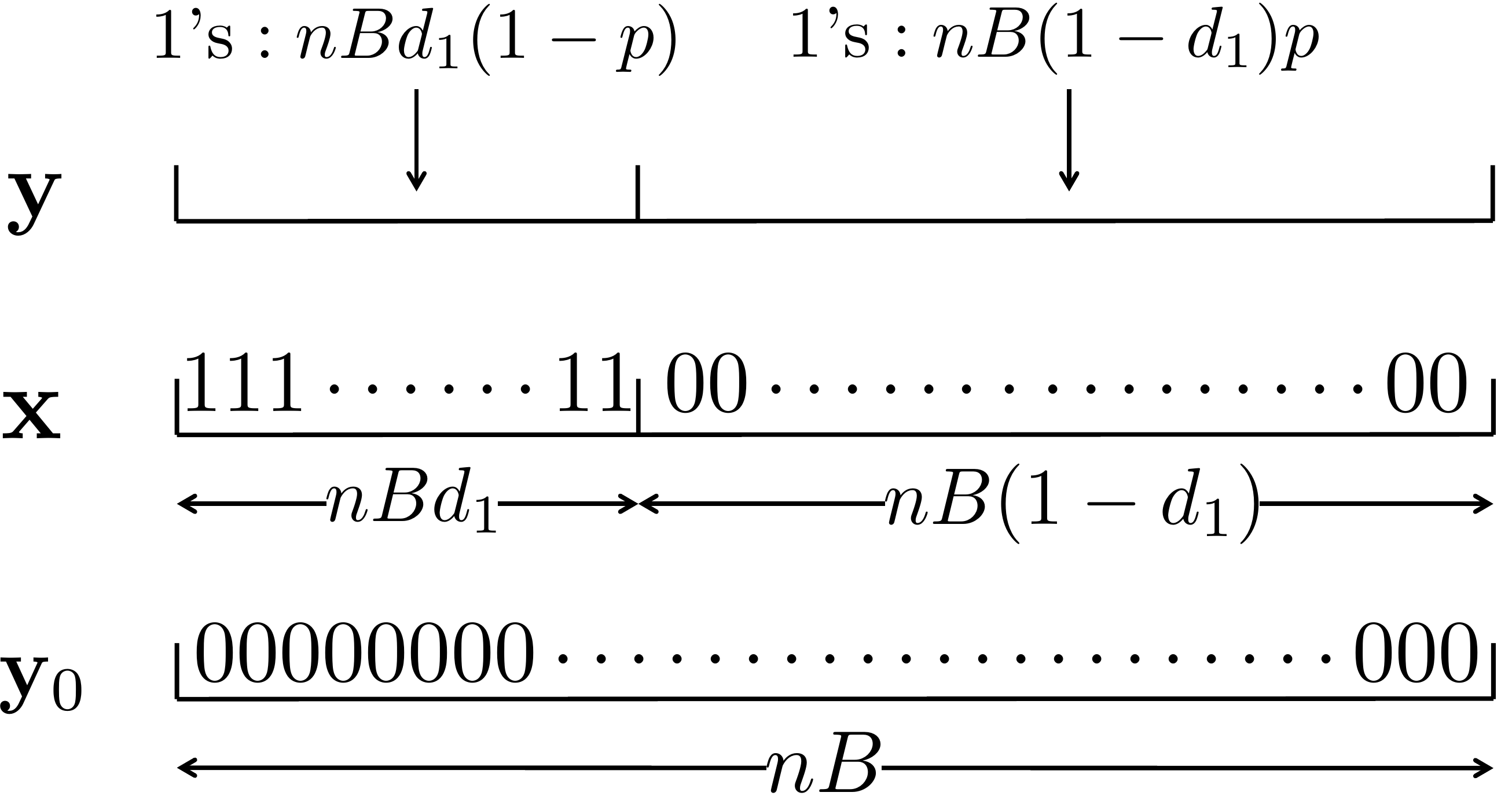}
\caption{Illustration of a specific pair $(\mathbf x,\mathbf y_0 )$.}
\label{F:xy_0}
\end{figure}

Combining \dref{E:qd2} and \dref{E:qd1}, we have for any $\delta>0$ and $B$ sufficiently large,
\begin{align*}
&~\mbox{Pr}(  d(\mathbf Y, \mathbf y_0) \geq  nB (d*p -\epsilon_0 ),p(\mathbf{Y}|\mathbf{z}_0)\leq 2^{-B(H(Y^n|Z^n)-\epsilon_0)}, d(\mathbf{Y},\mathbf{z}_0) \in [ nB(p*p- \epsilon_0),  nB(p*p+ \epsilon_0)]    |\mathbf z_0)\\
&\geq 1- ( {\delta}+1- q_d  (1-\delta))\\
&= q_d  (1-\delta) -  {\delta} \\
&\geq q_d  - 2 \delta.
\end{align*}

On the other hand,
\begin{align*}
&~\mbox{Pr}(  d(\mathbf Y, \mathbf y_0) \geq  nB (d*p -\epsilon_0 ),p(\mathbf{Y}|\mathbf{z}_0)\leq 2^{-B(H(Y^n|Z^n)-\epsilon_0)}, d(\mathbf{Y},\mathbf{z}_0) \in [ nB(p*p- \epsilon_0),  nB(p*p+ \epsilon_0)]    |\mathbf z_0)\\
&\leq 2^{-B(H(Y^n|Z^n)-\epsilon_0)} |\{\mathbf y: d(\mathbf y, \mathbf y_0) \geq  nB (d*p -\epsilon_0 ), d(\mathbf{y},\mathbf{z}_0) \in [ nB(p*p- \epsilon_0),  nB(p*p+ \epsilon_0)]    \}|\\
&= 2^{-B(H(Y^n|Z^n)-\epsilon_0)} \left| \Cup_{r=d*p-\epsilon_0}^{1}\mbox{Sphere}(\mathbf y_0, nBr ) \Cap   \Cup_{\rho=p*p- \epsilon_0}^{p*p+ \epsilon_0}\mbox{Sphere}(\mathbf z_0, nB\rho ) \right|\\
&= 2^{-B(H(Y^n|Z^n)-\epsilon_0)} \left| \Cup_{r=d*p-\epsilon_0}^{1}\Cup_{\rho=p*p- \epsilon_0}^{p*p+ \epsilon_0}\underbrace{\mbox{Sphere}(\mathbf y_0, nBr ) \Cap   \mbox{Sphere}(\mathbf z_0, nB\rho )}_{\mbox{Inter}(r,\rho)}\right|.
\end{align*}
Therefore,
\begin{align}
q_d\leq 2 \delta + 2^{-B(H(Y^n|Z^n)-\epsilon_0)} \left| \Cup_{r=d*p-\epsilon_0}^{1}\Cup_{\rho=p*p- \epsilon_0}^{p*p+ \epsilon_0} \mbox{Inter}(r,\rho) \right|. \label{E:vanish}
\end{align}

\begin{figure}[hbt]
\centering
\includegraphics[width=0.4\textwidth]{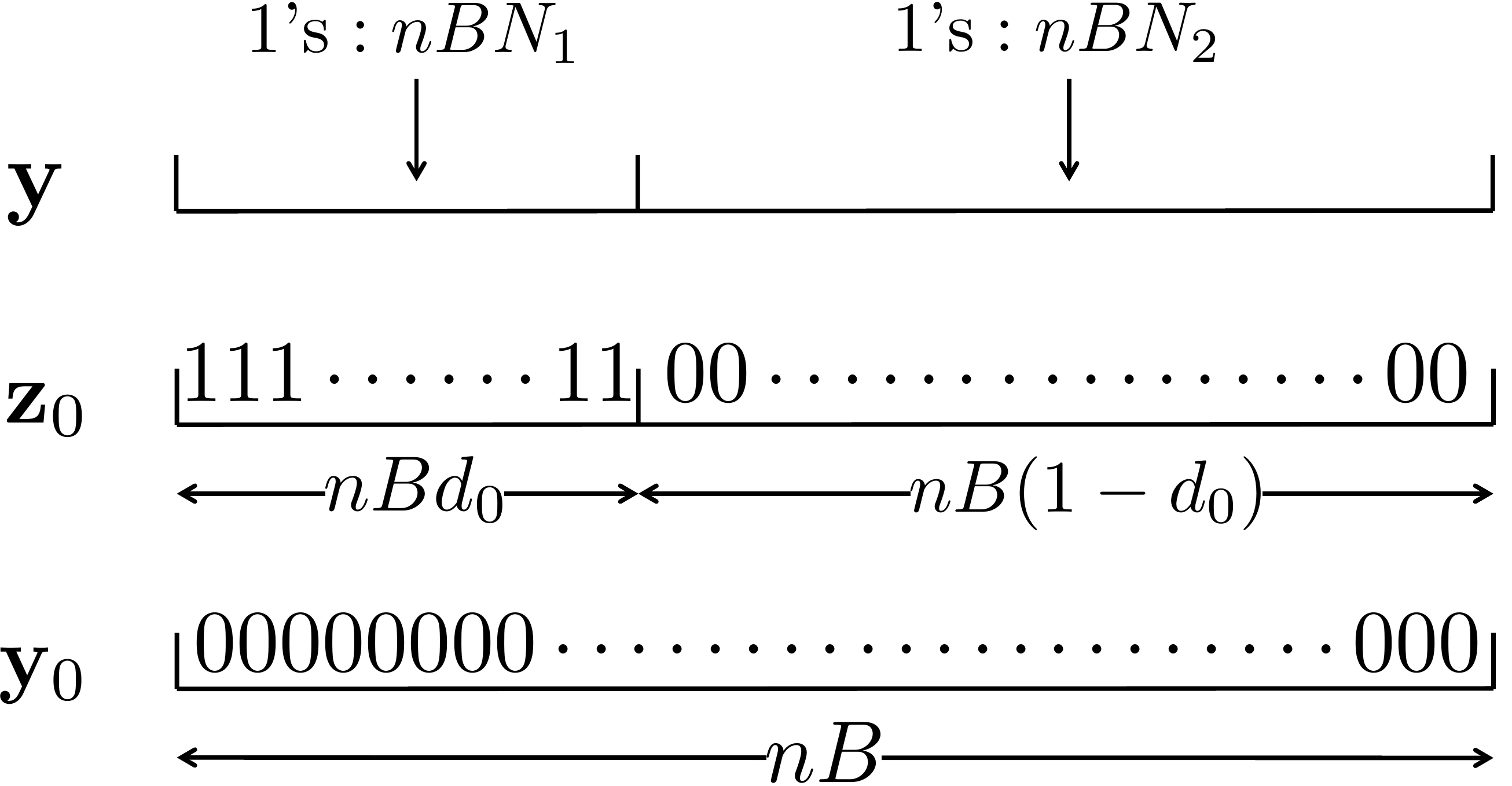}
\caption{Illustration of a specific pair $(\mathbf y_0,\mathbf z_0 )$.}
\label{F:z_0y_0}
\end{figure}

Now we show that the second term on the R.H.S. of \dref{E:vanish} vanishes if $d>d_0*p$.
Without loss of generality, consider a specific pair $(\mathbf y_0,\mathbf z_0 )$ as shown in Fig. \ref{F:z_0y_0},
where $\mathbf y_0=\mathbf 0$ and $d(\mathbf y_0,\mathbf z_0 )=nd_0$. We first characterize the volume of $\mbox{Inter}(r,\rho)$ for any $r$ and $\rho = p*p$, i.e.,
\begin{align}
\left| \mbox{Sphere}(\mathbf y_0, nBr ) \Cap   \mbox{Sphere}(\mathbf z_0, nB p*p ) \right| .\label{E:vol}
\end{align}
For any $\mathbf y$, let $nBN_1$ denote the number of 1's appearing in its first $nBd_0$ digits, and $nBN_2$ denote the number of 1's in the rest. Then the volume in \dref{E:vol} amounts to the number of $\mathbf y$'s such that the following two equalities hold
\begin{numcases}{}
 N_1+N_2= r \nonumber \\
 d_0-N_1+N_2=p*p\nonumber
\end{numcases}
i.e.,
\begin{numcases}{}
 N_1 = \frac{r+d_0-p*p}{2}\nonumber\\
 N_2 = \frac{r+p*p-d_0}{2}\nonumber.
\end{numcases}
Here, $N_1\in [0,d_0]$ and $N_2\in [0, (1-d_0)]$, i.e.,
\begin{numcases}{}
 r \in [p*p-d_0,p*p+d_0 ] \label{E:rcondition1}\\
 r \in [d_0-p*p,2-p*p-d_0 ]\label{E:rcondition2}.
\end{numcases}
Therefore, with $\rho=p*p$,
\begin{align}
\mbox{Inter}(r,\rho)&= { nBd_0  \choose nB\frac{r+d_0-p*p}{2}}{ nB(1-d_0)  \choose nB\frac{r+p*p-d_0}{2}}\nonumber \\
&\leq 2^{nBd_0 \left(H( \frac{r+d_0-p*p}{2d_0} ) \right)} 2^{nB(1-d_0) \left(H( \frac{r+p*p-d_0}{2(1-d_0)} ) \right)} \label{E:stirling}\\
&= 2^{nB \left(d_0 H( \frac{r+d_0-p*p}{2d_0} ) +(1-d_0) H( \frac{r+p*p-d_0}{2(1-d_0)} )\right)}\nonumber
\end{align}
for sufficiently large $B$, where \dref{E:stirling} follows from the bound ${n \choose nk} \leq \frac{1}{\sqrt{n\pi k(1-k)}} 2^{nH(k)}$ for any $k \in (0,1)$, as stated in \cite[Lemma 17.5.1]{Coverbook}.

Let $f(r)=d_0 H( \frac{r+d_0-p*p}{2d_0} ) +(1-d_0) H( \frac{r+p*p-d_0}{2(1-d_0)} )$.
%Obviously $f(r)$ is continuous, and
It can be verified that $f(r)$ attains the maximum $H(p*p)$ if and only if $r=d_0*p*p$; see Appendix \ref{A:f(r)}. Thus,  when
\begin{align*}
\rho&=p*p\\
d&=d_0*p+\epsilon'\\
r&\geq d*p-\epsilon_0=(d_0*p+\epsilon')*p-\epsilon_0=d_0*p*p+\epsilon'(1-2p)-\epsilon_0,
\end{align*}
we have
\begin{align*}
\mbox{Inter}(r,\rho)\leq 2^{nB(H(p*p)-\epsilon_1 )}
\end{align*}
for some $\epsilon_1>0$ provided $\epsilon'(1-2p)-\epsilon_0>0$. Further, due to the continuity of $\mbox{Inter}(r,\rho)$ in $\rho$, for any $\rho \in [p*p-\epsilon_0, p*p+\epsilon_0]$, $d=d_0*p+\epsilon'$,  $r\geq d*p-\epsilon_0$,
\begin{align}
\mbox{Inter}(r,\rho)\leq 2^{nB(H(p*p)-\epsilon_1+\epsilon_2)}\label{E:pluggingtovanish}
\end{align}
for some $\epsilon_2\to 0$ as $\epsilon_0 \to 0$.

Plugging \dref{E:pluggingtovanish} into \dref{E:vanish}, we have for $d=d_0*p+\epsilon'$ and sufficiently large $B$,
\begin{align*}
q_d&\leq 2 \delta+ 2^{-B(H(Y^n|Z^n)-\epsilon_0)}  \sum_{r=d*p-\epsilon_0}^{1}\sum_{\rho=p*p-\epsilon_0}^{p*p+\epsilon_0}
2^{nB(H(p*p)-\epsilon_1 +\epsilon_2)}   \\
&\leq  2 \delta+ 2^{-B(H(Y^n|Z^n)-\epsilon_0)}
2^{nB(H(p*p)-\epsilon_1+\epsilon_0 +\epsilon_2)}  \\
&\leq  2 \delta+ 2^{-nB(\frac{1}{n}H(Y^{n}|Z^{n}) -H(p*p)  +\epsilon_1   -2\epsilon_0-\epsilon_2)}   ,
\end{align*}
with
\begin{align}
\frac{1}{n}H(Y^{n}|Z^{n})&=\frac{1}{n}(H(Y^n,Z^n)-H(Z^n)) \nonumber\\
&=\frac{1}{n}(H(X^n)+H(Y^n,Z^n|X^n)-H(X^n|Y^n,Z^n)-H(Z^n)) \nonumber\\
&=\frac{1}{n}(H(M)-H(M|X^n)+H(Y^n,Z^n|X^n)-H(X^n|Y^n,Z^n)-H(Z^n)) \nonumber\\
&\geq \frac{1}{n}(nR-n\epsilon_0+2nH(p)-n\epsilon_0 -n) \label{E:lastfano}\\
&= C_{XYZ}-\tau +2H(p)-1-2\epsilon_0 \nonumber\\
&=  H(p*p)-\tau -2\epsilon_0, \nonumber
\end{align}
for $n$ sufficiently large, where in \dref{E:lastfano} we have used Fano's inequality.
% to obtain that $H(M|X^n)\leq H(M|Y^n,Z^n)\leq n\epsilon_0$
Thus, when $\tau,\epsilon_0$ are sufficiently small and $n,B$ are sufficiently large, we have for any $\epsilon'>0$, $d=d_0*p+\epsilon'$,
\begin{align*}
q_{d} &\leq 2 \delta+ 2^{-nB(\epsilon_1 -  \tau -4\epsilon_0-\epsilon_2)} \\
&\leq 3 \delta.
\end{align*}
We finally conclude that for a $(\mathbf{y}_0,\mathbf{z}_0)$ pair where $\mathbf{z}_0\in \mathcal T_{ \epsilon}^{(B)}(Z^n)$ and $d(\mathbf{y}_0,\mathbf{z}_0)=nBd_0$, there exists some $\epsilon'\to 0$ as $\epsilon\to 0$ such that
\begin{align*}
\mbox{Pr}( d(\mathbf X, \mathbf y_0) \leq  nB ( d_0*p+\epsilon')   |\mathbf z_0)\geq 1-\upsilon .
\end{align*}
where $\upsilon$ can be made arbitrarily small by choosing $n,B$ sufficiently large and $\tau$ sufficiently small.

\section{Conclusion}\label{conclusion}
We consider the symmetric primitive relay channel, and develop two new upper bounds on its capacity that are tighter than existing bounds, including the celebrated cut-set bound. Our approach uses measure concentration (the blowing-up lemma in particular) to analyze the probabilistic geometric relations between the typical sets of the  $n$-letter random variables associated with a reliable code for communicating over this channel. We then translate these relations  to new entropy inequalities between the $n$-letter random variables involved.

Information theory and geometry are indeed known to be inherently related; for example the differential entropy of a continuous random variable can be
regarded as the exponential growth rate of the volume of its typical set. Therefore, entropy relations can, in principle, be developed by studying the relative geometry of the typical sets of the random variables. However, we are not aware of many examples where such geometric techniques have been successfully used to develop converses for problems in network information theory. It would be interesting to see if the approach we develop in this paper, i.e. deriving information inequalities
by studying the geometry of typical sets, in particular using measure concentration,  can be used to make progress on other long-standing open problems in network information theory.

While we have exclusively focused on the symmetric relay channel in this paper, our results can be extended to asymmetric primitive relay channels \cite{ISIT2016} using the idea of channel simulation. An extension of these ideas to the Gaussian case has been provided in \cite{Allerton2015}. %As an application of our bounds, we consider the binary symmetric channel case and study the question of what is the minimum required relay--destination communication rate $R^*_0$ so as to achieve the capacity $C_{XYZ}$ of the broadcast cut. Surprisingly, we show that as $p\to 1/2$, $R_0^*\geq 0.1803$, implying that even as the broadcast channel becomes completely noisy, to achieve the diminishing $C_{XYZ}$ a strictly positive $R_0$ is needed. The exact characterization of $R^*_0$ still remains elusive.
\appendices

\section{}\label{A:reliability}
To see $E(R)\leq R-C_{XY}$ for any $R>I(X;Y)$, recall that $E(R)$ has the following alternative form \cite{Korner}:
\begin{align}
E(R)=   \min_{p(x)} \min_{\tilde p(y|x)} D( \tilde p(y|x)||p(y|x)|p(x)) + |R-I(p(x),\tilde p(y|x))| ^{+}       \label{E:Edefine1}
\end{align}
where $|t|^+:=\max\{0,t\}$, $D( \tilde p(y|x)||p(y|x)|p(x))$ is the conditional relative entropy defined as
\begin{align*}
D( \tilde p(y|x)||p(y|x)|p(x)):=  \sum_{(x,y) } p(x)\tilde p(y|x) \log \frac{\tilde p(y|x) }{p(y|x)} ,
\end{align*}
and $I(p(x),\tilde p(y|x))$ is the mutual information defined with respect to the joint distribution $p(x)\tilde p(y|x)$, i.e.,
\begin{align*}
I(p(x),\tilde p(y|x)):=  \sum_{(x,y) } p(x)\tilde p(y|x) \log \frac{\tilde p(y|x) }{\sum_{x}p(x)\tilde p(y|x) } .
\end{align*}
In the regime of $R>C_{XY}$, simply choosing the $p(x)$ and $\tilde p(y|x)$ in \eqref{E:Edefine1} to be capacity-achieving distribution $p^*(x)$ and $p(y|x)$ respectively would make the objective function equal to $R-C_{XY}$, and thus $E(R)\leq R-C_{XY}$.

\section{}\label{A:example}

We demonstrate the improvements of our bound in Theorem \ref{T:sharpen} over Xue's bound using the following simple example.

\begin{example}
Suppose both $X$-$Y$ and $X$-$Z$ links are the binary asymmetric channels as depicted in Fig. \ref{F:bac}, with parameters $p_1=0.01$ and $p_2=0.3$. For the input distribution
\begin{numcases}{p(x)=}
\alpha   & $x=0$ \nonumber \\
1-\alpha &  $x=1$  \nonumber
\end{numcases}
we have $$I(X;Y)=H ( \alpha(1-p_1) +(1-\alpha) p_2   )- (\alpha H(p_1)+(1-\alpha) H(p_2) )$$
and
\begin{align*}
I(X;Y,Z)=& \ H ( [  \alpha(1-p_1)^2+(1-\alpha)p_2 ^2,   \alpha(1-p_1)p_1+(1-\alpha)(1-p_2)p_2 ,  \\
            &~~~~~~~ \alpha(1-p_1)p_1+(1-\alpha)(1-p_2)p_2 ,   \alpha p_1^2+(1-\alpha)(1-p_2)^2 ]   )\\
&-  2(\alpha H(p_1)+(1-\alpha) H(p_2) ).
\end{align*}
With $p_1=0.01$ and $p_2=0.3$, numerical evaluation of $I(X;Y)$ and $I(X;Y,Z)$ yields that $C_{XY}=\max_{\alpha} I(X;Y)=0.46432$ with the maximizer $\alpha^* _{XY}=0.58$, and $C_{XYZ}=\max_{\alpha} I(X;Y,Z)=0.72022$ with the maximizer $\alpha^* _{XYZ}=0.54$.

Suppose we want to achieve a rate $R=C_{XYZ}$, and we use Proposition \ref{P:Xue} and Theorem \ref{T:sharpen} to derive a lower bound on $R_0$, respectively.
First consider Proposition \ref{P:Xue}. Numerically, we have $E(R)=0.05951$ for $R=C_{XYZ}=0.72022$, and the minimum $a$ to satisfy \dref{E:Xue2} is $a=0.00008$. Thus, by \dref{E:Xue1}, we have
\begin{align}
R_0&\geq R-C_{XY}+a \nonumber \\
&\geq 0.72022- 0.46432 + 0.00008\label{E:bothtobesharpened} \\
& = 0.25598.\nonumber
\end{align} We then apply Theorem \ref{T:sharpen} and demonstrate that the improvements mentioned in Section \ref{SS:Newbounds} result in a tighter bound on $R_0$.
In Theorem \ref{T:sharpen}, $p(x)$ has to be chosen such that $\alpha =\alpha^* _{XYZ}=0.54$ due to the constraint \dref{E:sharpened1}. Under such a distribution of $p(x)$, numerically, we have $I(X;Y)=0.46223 < C_{XY}$, and $C_{XYZ}-I(X;Y)= 0.25799 > E(R)$. Noting the R.H.S.  of  \dref{E:Xue2} is also sharpened to that of \dref{E:rewrite}, we can calculate the minimum $a$ satisfying \dref{E:rewrite} to be $a=0.00546$. Thus, by \dref{E:sharpened2}, we have
\begin{align}
R_0&\geq R-I(X;Y) +a \nonumber \\
&\geq 0.72022- 0.46223 + 0.00546 \label{E:bothsharpened}\\
& =0.26345,\nonumber
\end{align}
where it is easy to see that the last two terms in \dref{E:bothsharpened} are both sharpened compared to those in \dref{E:bothtobesharpened}.

Therefore, in order to achieve the rate $R=C_{XYZ}$, the lower bounds on $R_0$ yielded by Proposition \ref{P:Xue} and Theorem \ref{T:sharpen} are
$$R_0 \geq 0.25598$$ and $$R_0\geq 0.26345$$ respectively. Viewed from another perspective, for $R_0 \in [0.25598, 0.26345 ) $, the bound in Theorem \ref{T:sharpen}
asserts that the capacity of the relay channel $C(R_0)<C_{XYZ}$, which excludes the possibility of achieving $R=C_{XYZ}$ while Xue's bound in Proposition \ref{P:Xue} cannot.

\begin{figure}[hbt]
\centering
\includegraphics[width=0.35\textwidth]{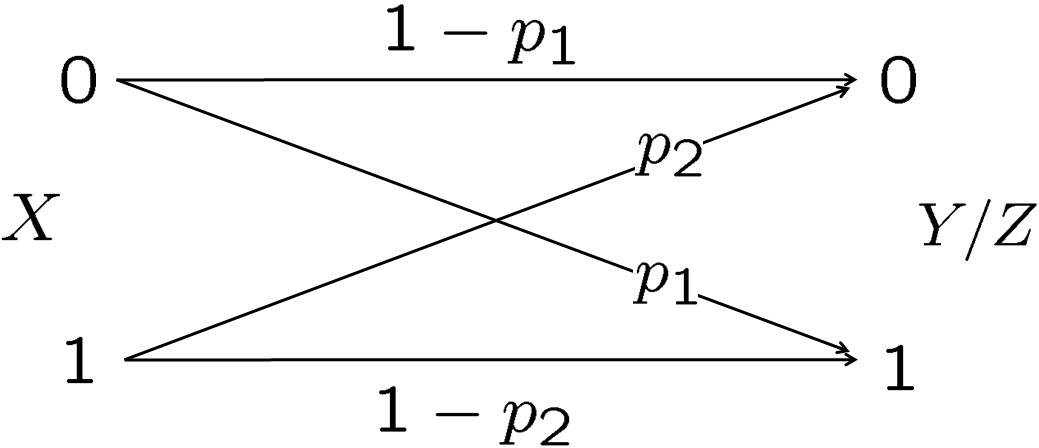}
\caption{Binary asymmetric channel.}
 \label{F:bac}
\end{figure}

\end{example}

\section{$\Delta\left(p(x),d\right)$ for Binary Symmetric Channels}\label{A:derivebscdelta}
For a binary symmetric channel with crossover probability $p<1/2$, the objective function in \dref{E:objective_def} can be expressed as
\begin{align}
&H(\tilde p(\omega|x) |p(x))+D( \tilde p(\omega|x) || p(\omega|x) |p(x))- H(p(\omega|x) |p(x))\nonumber \\
=\ &-\sum_{(x,\omega) } p(x)\tilde p(\omega|x) \log \tilde p(\omega|x)
+ \sum_{(x,\omega) }  p(x)\tilde p(\omega|x) \log \frac{\tilde p(\omega|x) }{p(\omega|x)} -\sum_{(x,\omega)  } p(x)  p(\omega|x) \log \frac{1}{p(\omega|x)}\nonumber  \\
%=\ &  \sum_{(x,\omega) }  p(x)\tilde p(\omega|x) \log \frac{1}{p(\omega|x)} -\sum_{(x,\omega)  } p(x)  p(\omega|x) \log \frac{1}{p(\omega|x)}  \\
=\ &  \sum_{(x,\omega) }  [p(x)\tilde p(\omega|x)-p(x)  p(\omega|x) ]\log \frac{1}{p(\omega|x)}\nonumber \\
=\ &  \sum_{x\neq \omega }  \left[p(x)\tilde p(\omega|x)-p(x)  p(\omega|x) \right]\log \frac{1}{p} +\sum_{x = \omega }  [p(x)\tilde p(\omega|x)-p(x)  p(\omega|x) ]\log \frac{1}{1-p} \label{E:already1}\\
=\ &  \sum_{x\neq \omega }  \left[p(x)\tilde p(\omega|x)-p(x)  p(\omega|x) \right]\log \frac{1}{p} +\sum_{x \neq \omega }  [p(x)(1-\tilde p(\omega|x))-p(x)  (1-p(\omega|x)) ]\log \frac{1}{1-p}\label{E:already2}\\
=\ &  \sum_{x\neq \omega }  \left[p(x)\tilde p(\omega|x)-p(x)  p(\omega|x) \right]\log \frac{1}{p} -\sum_{x \neq \omega }  [p(x)\tilde p(\omega|x)-p(x) p(\omega|x) ]\log \frac{1}{1-p}\label{E:already3}\\
=\ &  \sum_{x\neq \omega }  \left[p(x)\tilde p(\omega|x)-p(x)  p(\omega|x) \right]\log \frac{1-p}{p}\label{E:newobjective}
\end{align}

We now show that under the constraint 
\begin{align}
\frac{1}{2}\sum_{(x,\omega)}|p(x)\tilde p(\omega|x)-p(x)  p(\omega|x)|\leq  d,\label{E:constraint_defApendix}
\end{align}
the function in \dref{E:newobjective}, and thus $\Delta\left(p(x),d\right)$, are upper bounded by $$\min\left\{ d,1-p   \right\} \log \frac{1-p}{p}.$$
%Since $\log \frac{1-p}{p}>0$, it suffices to show that
%$\sum_{x\neq \omega }  \left[p(x)\tilde p(\omega|x)-p(x)  p(\omega|x) \right] \leq \min\{ \sqrt{a/2},1-p   \}$. For this, we first derive a constraint equivalent to
%\dref{E:constraint_def}.
Along the similar lines as in \dref{E:already1}--\dref{E:already3}, we obtain
$$\sum_{x = \omega } |p(x)\tilde p(\omega|x)-p(x)  p(\omega|x) |=\sum_{x \neq \omega } |p(x)\tilde p(\omega|x)-p(x)  p(\omega|x) |,$$
and thus the constraint  \dref{E:constraint_defApendix} can be rewritten as
\begin{align}
\sum_{x \neq \omega } |p(x)\tilde p(\omega|x)-p(x)  p(\omega|x) |\leq d.\label{E:constraint_part1}
\end{align}
On the other hand,  we have
\begin{align}
\sum_{x \neq \omega } |p(x)\tilde p(\omega|x)-p(x)  p(\omega|x) |&=\sum_{x \neq \omega } p(x)|\tilde p(\omega|x)-  p(\omega|x) |\nonumber \\
&\leq \sum_{x \neq \omega }  p(x) (1-p) \nonumber \\
&= 1-p.\label{E:constraint_part2}
\end{align}
Combining \dref{E:newobjective}, \dref{E:constraint_part1} and \dref{E:constraint_part2} yields that
$$\Delta\left(p(x),d\right) \leq   \min\left\{ d,1-p   \right\} \log \frac{1-p}{p}.  $$
In fact, it can be easily checked that the equality sign in the above inequality can be attained by choosing
\begin{align*}
\tilde p(\omega|x)=p+\min\left\{ d,1-p   \right\}, \forall (x,\omega)
\text{ with } x\neq\omega,
\end{align*}
 and thus we conclude that
$$\Delta\left(p(x),d\right)= \min\left\{d,1-p   \right\}\log \frac{1-p}{p}.  $$

\section{Upper Bounds for Binary Symmetric Channel Case}\label{A:BSCBounds}

Various upper bounds are evaluated for the binary symmetric channel case as follows.
\subsection{Cut-Set bound (Prop. \ref{P:cutset})}
The optimal distribution for Prop. \ref{P:cutset} is $p^*(0)=p^*(1)=1/2$, under which,
\begin{numcases}{}
 I^*(X;Y,Z) = C_{XYZ}=1+H(p*p)-2H(p) \nonumber \\
 I^*(X;Y)=C_{XY}= 1 - H(p) \nonumber.
\end{numcases}
Therefore, the cut-set bound simplifies to
  \begin{align*}
  C(R_0) \leq  \min \left\{1+H(p*p)-2H(p), 1-H(p)+R_0   \right\}.
  \end{align*}

\subsection{Xue's bound (Prop. \ref{P:Xue})}
Since the function $E(R)$ is monotonic in $R$, its inverse function $E^{-1}(\cdot)$ exists and Xue's bound can be expressed as
\begin{align*}
  C(R_0) \leq \max_{a\in [0,R_0] } \min \left\{1-H(p)+R_0-a, E^{-1}(H(\sqrt{a})+  \sqrt{a} )   \right\}.
  \end{align*}

\subsection{Our first bound (Thm. \ref{T:sharpen})}
The uniform distribution of $X$ is also optimal for Thm. \ref{T:sharpen}, under which our first bound reduces to
  \begin{align*}
  C(R_0) \leq \max_{a\in [0, \min \{R_0, H(p),  \frac{1}{2\ln 2}\}] }  \min \Bigg \{1+H(p*p)-2H(p), \  &   1-H(p)+R_0-a ,    \\
&   1-H(p)+  H\left(\sqrt{\frac{a \ln 2}{2}}\right)-a         \Bigg\}.
  \end{align*}
where the constraint of $a$ follows since $H(Z|X)=H(p)$ for any $p(x)$ and $\frac{2}{\ln 2} \left( \frac{|\Omega|-1}{|\Omega|} \right)^2= \frac{1}{2\ln 2}$, and the term $\sqrt{\frac{a \ln 2}{2}}\log (|\Omega|-1)$ disappears compared to Thm. \ref{T:sharpen} since it becomes 0 with $|\Omega|=2$.

\subsection{Our second bound (Thm. \ref{T:novel})}
Recall that in the binary symmetric channel case, $\Delta\left(p(x),d\right)$ is independent of $p(x)$ and given by $$\min\left\{ d,1-p   \right\} \log \frac{1-p}{p}:= \bar{\Delta}\left(d\right).$$
Therefore, our new bound becomes
\begin{align}
  C(R_0) \leq \max_{a\in [0, \min \{R_0, H(p)\}] }  \min \left\{1+H(p*p)-2H(p) ,    1-H(p)+R_0-a ,   1-H(p)+\bar{\Delta}\left(\sqrt{\frac{a \ln 2}{2}}\right)\right\}.\label{E:rhs}
  \end{align}
It is not difficult to see that for the R.H.S. of \dref{E:rhs}, at least one of the maximizers must be no greater than
$\frac{2}{\ln 2}(1-p)^2$, i.e., satisfying $\sqrt{\frac{a \ln 2}{2}}\leq 1-p$. Therefore, \dref{E:rhs} can be equivalently stated as
\begin{align*}
  C(R_0) \leq \max_{a\in [0, \min \{R_0, H(p),\frac{2}{\ln 2}(1-p)^2\}] }  \min \Bigg \{1+H(p*p)-2H(p)&,     1-H(p)+R_0-a ,    \\
&   1-H(p)+\sqrt{\frac{a \ln 2}{2}} \log \frac{1-p}{p} \Bigg\}.
  \end{align*}

\section{Proof of \dref{E: HFopt}}\label{A:optimality}
To show \dref{E: HFopt}, it suffices to show $H(p)/H(p*p)\to 1/2$ as $p \to 0$. For this, we have
\begin{align*}
\lim_{p\to 0}\frac{H(p)}{H(p*p)}&=\lim_{p\to 0}\frac{H'(p)}{H'(p*p)\cdot (p*p)'}\\
&=\lim_{p\to 0}\frac{\log \frac{1-p}{p}}{\log \frac{1-p*p}{p*p}\cdot (2-4p)}\\
&=\frac{1}{2}\lim_{p\to 0}\frac{\log' \frac{1-p}{p} \cdot(\frac{1-p}{p})'}{\log' \frac{1-p*p}{p*p} \cdot (\frac{1-p*p}{p*p})'}\\
&=\frac{1}{2}\lim_{p\to 0}\frac{  \frac{ p}{1-p} \cdot (-\frac{1}{p^2})}{  \frac{ p*p}{1-p*p} \cdot (-\frac{1}{(p*p)^2})\cdot (2-4p)}\\
&=\frac{1}{4}\lim_{p\to 0}\frac{(p*p) (1-p*p) }{p(1-p)}\\
&=\frac{1}{4}\lim_{p\to 0}2 (1-p*p)  \\
&=\frac{1}{2}   \end{align*}
which proves \dref{E: HFopt}.

\section{Proof of Lemma \ref{L:fixed}}\label{A:fixedproof}
To prove Lemma \ref{L:fixed}, we need the following lemma, whose proof relies on the property of polynomial number of compositions and can be easily extended from the proof
for a single user channel \cite{Gallager}.
%
%postponed until we prove Lemma \ref{L:fixed}.

\begin{lemma}\label{L:fixedproof}
Suppose a code $ (\mathcal{C}_{(n,R)}, f_n,g_n     )$ has the probability of error $P_e^{(n)}$. Then there is some composition $Q$ for
which a fixed composition code $ (\mathcal{C}_{(n, \tilde R)}^{[Q]}, f_n,g_n     )$ exists, and its probability of error $\tilde P_e^{(n)}$
and rate $\tilde R$ satisfy $\tilde P_e^{(n)} \leq  P_e^{(n)}$ and $\tilde R\geq R-(|\Omega_X|-1) \frac{\log (n+1)}{n} $.
\end{lemma}

\begin{proof}[Proof of Lemma \ref{L:fixed}]
Since $R$ is achievable, there exists a sequence of codes $$\{(\mathcal{C}_{(n,R)}, f_n, g_n)\}_{n=1}^{\infty}$$
such that $P_e^{(n)} \to 0$ as $n \to \infty$. By Lemma \ref{L:fixedproof}, there accordingly exists a sequence of   fixed composition codes
$$\{(\mathcal{C}_{(n,\tilde R)}^{[Q_n]}, f_n, g_n)\}_{n=1}^{\infty}$$
such that the probability of error $\tilde P_e^{(n)} \leq  P_e^{(n)}$ and the rate $\tilde R\geq R-(|\Omega_X|-1) \frac{\log (n+1)}{n}$. Noting that
$\tilde R \to R$ and $\tilde P_e^{(n)}\to 0$ as $n\to \infty$, Lemma \ref{L:fixed} is thus proved.
\end{proof}

%
%\section{Proof of Proposition \ref{computation}}\label{A:delta}
%Suppose $J_1>J_2$. The optimal choice of $\lambda_{\pi(i)}$ would be
%\begin{numcases}{\lambda_{\pi(i) }=}
% P_{\pi(i) } & if $i\leq J_2$\nonumber \\
%0 & otherwise \nonumber
%\end{numcases}
%which leads to
%\begin{align*}
%\Delta(p(x),a)=\sum_{i=1}^{J_2} P_{\pi(i) }\gamma_{\pi(i) }.
%\end{align*}
%
%Suppose $J_1\leq J_2$. The optimal choice of $\lambda_{\pi(i)}$ would be
%\begin{numcases}{\lambda_{\pi(i) }=}
% P_{\pi(i) } & if $i\leq J_1-1$\nonumber \\
% \min\{\sqrt{a/2},1\}- \sum_{i=1}^{J_1-1}  P_{\pi(i) } & if $i = J_1$\nonumber \\
%   0 & if $i > J_1$\nonumber
%\end{numcases}
%and
%\begin{align*}
%\Delta(p(x),a)= \sum_{i=1}^{J_1} P_{\pi(i) }\gamma_{\pi(i) } -   \left(\sum_{i=1}^{J_1}  P_{\pi(i) } -\sqrt{a/2}\right)^+   \gamma_{\pi(J_1) }          .
%\end{align*}
%Proposition \ref{computation} is proved by combining the above two cases.
%
%

\section{Proof of Lemma \ref{L:entropy}}\label{A:entropy}
We only characterize $H(Y^n|X^n)$ and $I(X^n;Y^n)$ while the other information quantities can be characterized similarly.
Consider $H(Y^n|x^n)$ for any specific $x^n$ with composition $Q_n$. We have
\begin{align*}
H(Y^n|x^n)&=\sum_{i=1}^n H(Y_i|x^n,Y^{i-1})\\
&=\sum_{i=1}^n H(Y_i|x_i)\\
&=\sum_{x} n Q_n(x) H(Y|x)\\
&=n H(Y|X)
\end{align*}
where $H(Y|X)$ is calculated based on $Q_n(x)p(\omega|x)$. Therefore, for the code with fixed composition $Q_n$,
\begin{align*}
H(Y^n|X^n)=\sum_{x^n}p(x^n)H(Y^n|x^n)=\sum_{x^n}p(x^n)(n H(Y|X))=n H(Y|X).
\end{align*}

To bound $I(X^n;Y^n)$, it suffices to bound $H(Y^n)$. For any specific $x^n$ with composition $Q_n$, we have for some $\epsilon_1\to 0$ as $n\to \infty$,
\begin{align*}
\mbox{Pr}( Y^n \in \mathcal T^{(n)}_{\epsilon_1}(Y)|x^n)\geq 1-\epsilon_1,
\end{align*}
where $\mathcal T^{(n)}_{\epsilon_1}(Y)$ is the typical set \cite{ElGamalKim} with respect to $\sum_{x}Q_n(x)p(\omega|x)$.
Therefore, for the code with fixed composition $Q_n$,
\begin{align*}
\mbox{Pr}( Y^n \in \mathcal T^{(n)}_{\epsilon_1}(Y))=\sum_{x^n}p(x^n)\mbox{Pr}( Y^n \in \mathcal T^{(n)}_{\epsilon_1}(Y)|x^n)\geq 1-\epsilon_1.
\end{align*}
Letting
$
W= \mathbb{I}( Y^n \in \mathcal T^{(n)}_{\epsilon_1}(Y)),
$
we have
\begin{align*}
H(Y^n)&\leq H(Y^n,W)\\
&\leq 1+H(Y^n|W)\\
&=1+ \mbox{Pr}( Y^n \in \mathcal T^{(n)}_{\epsilon_1}(Y)) H(Y^n|Y^n \in \mathcal T^{(n)}_{\epsilon_1}(Y))+
\mbox{Pr}( Y^n \notin \mathcal T^{(n)}_{\epsilon_1}(Y)) H(Y^n|Y^n \notin \mathcal T^{(n)}_{\epsilon_1}(Y))\\
&\leq 1+  \log |\mathcal T^{(n)}_{\epsilon_1}(Y)| + n \epsilon_1 \log|\Omega| \\
&\leq 1+  n(H(Y)+\epsilon_2) + n \epsilon_1 \log|\Omega|\\
&\leq n(H(Y)+\epsilon)
\end{align*}
where $\epsilon_1, \epsilon_2,\epsilon \to 0$ as  $n\to \infty$, and  $H(Y)$ is calculated based on $\sum_{x}Q_n(x)p(\omega|x)$. Combining this with the fact that $H(Y^n|X^n)=nH(Y|X)$, we have
\begin{align}
I(X^n;Y^n)\leq n(I(X;Y)+\epsilon).\label{E:tobedrop}
\end{align}

We now argue that the $\epsilon$ in \dref{E:tobedrop} can be dropped.  Given any fixed $n$, consider a length-$B$ sequence of i.i.d. random vector pairs $\{  (X^n(b),Y^n(b))\}_{b=1}^{B}$, denoted by $(\mathbf X, \mathbf Y)$, where $(X^n(b),Y^n(b))$ have the same distribution as $(X^n,Y^n)$ for any $b\in [1:B]$. Obviously the  length-$nB$ vector $\mathbf X$ also has composition $Q_n$, and  by \dref{E:tobedrop} we have
\begin{align*}
I(\mathbf X;\mathbf Y)\leq nB(I(X;Y)+\epsilon),
\end{align*}
where $\epsilon \to 0$ as $B \to \infty$. Due to the i.i.d. property, we further have
\begin{align*}
B I(X^n;Y^n) \leq nB(I(X;Y)+\epsilon).
\end{align*}
Dividing $B$ at both sides of the above equation and letting $B\to 0$, we obtain $I(X^n;Y^n) \leq nI(X;Y)$.

\section{Volume of a Hamming Ball}\label{A:ball}
Consider the volume of a general $n$-dimensional Hamming ball in $\Omega^n$ with radius $nr$. It is obvious that when $r\geq 1$, the volume
\begin{align}|\text{Ball}(nr)|=|\Omega|^n=2^{n \log |\Omega|}.\label{E:volume1}\end{align}
For $r<1$, we have
\begin{align}
|\text{Ball}(nr)|&=1+\sum_{k=\frac{1}{n}}^{r} {n\choose nk}(|\Omega|-1)^{{nk}}\nonumber \\
&\leq 1+\sum_{k=\frac{1}{n}}^{r} \frac{ 2^{nH(k)}}{\sqrt{\pi n k (1-k)}} (|\Omega|-1)^{nk}\nonumber \\
&\leq  \sum_{k=0}^{r}   2^{n(H(k)+    k\log (|\Omega|-1) )}  \nonumber \\
&\leq  (nr+1)\max_{k\in \{0,\frac{1}{n},\ldots,r\}}   2^{n(H(k)+    k\log (|\Omega|-1))}  \nonumber \\
&\leq   \max_{k\in \{0,\frac{1}{n},\ldots,r\}}   2^{n(H(k)+    k\log (|\Omega|-1) + \epsilon )}\nonumber \\
&=     2^{n\left(\max_{k\in \{0,\frac{1}{n},\ldots,r\}}  H(k)+    k\log (|\Omega|-1)  + \epsilon \right)}\nonumber
\end{align}
for any $\epsilon> 0$ and sufficiently large $n$, where the first inequality follows from \cite[Lemma 17.5.1]{Coverbook}. Similarly, we can lower bound $|\text{Ball}(nr)|$ as
\begin{align}
|\text{Ball}(nr)|\geq    2^{n\left(\max_{k\in \{0,\frac{1}{n},\ldots,r\}}  H(k)+    k\log (|\Omega|-1)  - \epsilon \right)} \nonumber
\end{align}
for any $\epsilon>0$ and sufficiently large $n$. Therefore,  we have for $r<1$,
\begin{align}
 \lim_{n\to \infty} \frac{1}{n}\log |\mbox{Ball}(nr)|=
 \max_{t\in [0,r]}H(t)+t\log(|\Omega|-1)  .\label{E:volume2}
 \end{align}

 Now we simplify the above expression. Let $v(t)=H(t)+t\log(|\Omega|-1)$ for $t\in (0,1)$. We have $v'(t)=\log\frac{(1-t)(|\Omega|-1 )}{t}$, which is decreasing in $t$ for $t\in (0,1)$ and equals 0 when $t=\frac{|\Omega|-1}{|\Omega|}$. Thus, the maximum of $v(t)$ is attained when $t=\frac{|\Omega|-1}{|\Omega|}$, and is given by
 \begin{align}
v^*(t)&=  H\left(\frac{|\Omega|-1}{|\Omega|}\right)+\frac{|\Omega|-1}{|\Omega|}\log(|\Omega|-1) \nonumber \\
&= - \frac{|\Omega|-1}{|\Omega|} \log\frac{|\Omega|-1}{|\Omega|} -\frac{ 1}{|\Omega|}\log\frac{ 1}{|\Omega|} + \frac{|\Omega|-1}{|\Omega|}\log(|\Omega|-1) \nonumber \\
&= \frac{|\Omega|-1}{|\Omega|}\log\frac{(|\Omega|-1) |\Omega|}{(|\Omega|-1)}    -\frac{ 1}{|\Omega|}\log\frac{ 1}{|\Omega|}        \nonumber \\
&=  \frac{|\Omega|-1}{|\Omega|}\log  |\Omega| + \frac{ 1}{|\Omega|}\log |\Omega| \nonumber \\
&=\log |\Omega|.\label{E:whichcase}
 \end{align}
 Therefore, we have
 \begin{numcases}{ \max_{t\in [0,r]}H(t)+t\log(|\Omega|-1)=}
\log |\Omega| & when $r\in \left(\frac{|\Omega|-1}{|\Omega|},1\right) $  \label{E:volume3} \\
H(r)+r\log(|\Omega|-1) &  when $r \in \left[0, \frac{|\Omega|-1}{|\Omega|}\right] $.\label{E:volume4}
\end{numcases}

Combining \dref{E:volume1}, \dref{E:volume2}, \dref{E:volume3} and \dref{E:volume4}, we obtain that
 \begin{numcases}{ \lim_{n\to \infty} \frac{1}{n}\log |\mbox{Ball}(nr)|=}
\log |\Omega| & when $r>  \frac{|\Omega|-1}{|\Omega|}  $  \nonumber \\
H(r)+r\log(|\Omega|-1) &  when $r \leq \frac{|\Omega|-1}{|\Omega|}  $.\nonumber
\end{numcases}

\section{}\label{A:jointlytypical}
For any $(\mathbf x,\mathbf z)  \in \mathcal T_{\epsilon}^{(B)}(X^n, Z^n)$, we have
\begin{align*}
|P_{\mathbf{x} }(x^n )- p(x^n ) |&\leq \epsilon   p(x^n ), \forall  x^n \\
|P_{(\mathbf{x},\mathbf{z})}(x^n,z^n)- p(x^n,z^n) |&\leq \epsilon   p(x^n,z^n), \forall (x^n,z^n)
\end{align*}
and thus
\begin{align}
|P_{\mathbf{z}|\mathbf{x}}(z^n |x^n)- p(z^n |x^n) |&\leq \epsilon_1   p(z^n |x^n), \forall (x^n,z^n) \text{ with } P_{\mathbf{x} }(x^n )\neq 0,
\end{align}
for some $\epsilon_1 \to 0$ as $\epsilon \to 0$.

Therefore, we have for any $(x,\omega)$ that
\begin{align*}
P_{(\mathbf{x},\mathbf{z})}(x ,\omega )&= \frac{1}{nB} N(x ,\omega| \mathbf{x},\mathbf{z})\\
&= \frac{1}{nB} \sum_{ (x^n,z^n) } N(x^n ,z^n| \mathbf{x},\mathbf{z})\cdot N(x,\omega |x^n ,z^n)\\
&=\sum_{ (x^n,z^n) } P_{(\mathbf{x},\mathbf{z})}(x^n,z^n) \cdot P_{(x^n,z^n)}(x ,\omega )\\
&=\sum_{ x^n:  P_{\mathbf{x}}(x^n)>0}P_{\mathbf{x}}(x^n) \sum_{ z^n } P_{ \mathbf{z}|\mathbf{x} }(z^n |x^n) \cdot P_{(x^n,z^n)}(x ,\omega )\\
&\leq \sum_{ x^n:P_{\mathbf{x}}(x^n)>0}P_{\mathbf{x}}(x^n) \sum_{z^n } p(z^n |x^n)(1+\epsilon_1) \cdot P_{(x^n,z^n)}(x ,\omega )\\
&= (1+\epsilon_1)\sum_{ x^n:P_{\mathbf{x}}(x^n)>0 }P_{\mathbf{x}}(x^n) E[  P_{(x^n,Z^n)}(x,\omega ) ]\\
&= (1+\epsilon_1)\sum_{ x^n:P_{\mathbf{x}}(x^n)>0 }P_{\mathbf{x}}(x^n) E\left[  \frac{1}{n}N(x, \omega|x^n,Z^n) \right]\\
&= (1+\epsilon_1)\sum_{ x^n:P_{\mathbf{x}}(x^n)>0}P_{\mathbf{x}}(x^n) E\left[  \frac{1}{n}\sum_{j: x_j=x}\mathbb{I} (Z_j=\omega) \right]\\
&= (1+\epsilon_1)\sum_{ x^n:P_{\mathbf{x}}(x^n)>0 }P_{\mathbf{x}}(x^n) P_{x^n}(x )   p(\omega|x)\\
&= (1+\epsilon_1)P_{\mathbf{x}}(x )  p(\omega|x)\\
&= (1+\epsilon_1)Q_n(x)  p(\omega|x).
\end{align*}
Similarly,
\begin{align*}
P_{(\mathbf{x},\mathbf{z})}(x ,\omega ) &\geq (1-\epsilon_1)Q_n(x)  p(\omega|x), \forall (x,\omega),
\end{align*}
and thus
\begin{align*}
|P_{(\mathbf{x},\mathbf{z})}(x ,\omega ) -Q_n(x)  p(\omega|x)| \leq \epsilon_1 Q_n(x)  p(\omega|x) , \forall (x,\omega).
\end{align*}

\section{}\label{A:totalvariance}
For any $(\mathbf x, \mathbf y, \mathbf z)$, consider the total variation distance between $P_{(\mathbf{x},\mathbf{y})}(x ,\omega )$ and $P_{(\mathbf{x},\mathbf{z})}(x ,\omega )$.
We have
\begin{align*}
&~~~~nB\sum_{(x,\omega)}|P_{(\mathbf{x},\mathbf{y})}(x ,\omega )-P_{(\mathbf{x},\mathbf{z})}(x ,\omega )|\\
&=\sum_{(x,\omega)}\left|\sum_{i=1}^{nB} \mathbb{I}((x_i,y_i)=(x,\omega) )-\mathbb{I}((x_i,z_i)=(x,\omega) )\right|\\
&=\sum_{(x,\omega)}\left|\sum_{i:y_i = z_i} \mathbb{I}((x_i,y_i)=(x,\omega) )-\mathbb{I}((x_i,z_i)=(x,\omega) )+\sum_{i:y_i \neq z_i} \mathbb{I}((x_i,y_i)=(x,\omega) )-\mathbb{I}((x_i,z_i)=(x,\omega) )\right|\\
&=\sum_{(x,\omega)}\left|\sum_{i:y_i \neq z_i} \mathbb{I}((x_i,y_i)=(x,\omega) )-\mathbb{I}((x_i,z_i)=(x,\omega) )\right|\\
&\leq \sum_{(x,\omega)} \sum_{i:y_i \neq z_i} \mathbb{I}((x_i,y_i)=(x,\omega) )+\sum_{(x,\omega)} \sum_{i:y_i \neq z_i}\mathbb{I}((x_i,z_i)=(x,\omega) ) \\
&=2 d(\mathbf y, \mathbf z),
\end{align*}
i.e.,
\begin{align*}
\sum_{(x,\omega)}|P_{(\mathbf{x},\mathbf{y})}(x ,\omega )-P_{(\mathbf{x},\mathbf{z})}(x ,\omega )|\leq \frac{2}{nB} d(\mathbf y, \mathbf z) .
\end{align*}

\section{}\label{A:bscproperty}
From the property of jointly typical sequences, for any $\mathbf z_0 \in \mathcal T_{\epsilon}^{(B)}(Z^n)$ and $\epsilon_1>\epsilon$,
\begin{align}
\text{Pr}( (\mathbf X, \mathbf Y, \mathbf z_0)\in \mathcal T_{\epsilon_1}^{(B)}(X^n, Y^n, Z^n)|\mathbf z_0) \to 1 \text{ as } B\to \infty.\label{E:bsca1}
\end{align}

For any $( \mathbf y, \mathbf z_0)\in \mathcal T_{\epsilon_1}^{(B)}(Y^n, Z^n)$,
\begin{align}
p(\mathbf y| \mathbf z_0)\leq 2^{-B(H(Y^n|Z^n)-\epsilon_2 )}, \text{ for some } \epsilon_2\to 0 \text{ as }\epsilon_1\to 0.\label{E:bsca2}
\end{align}

Also, along the same lines as Appendix \ref{A:jointlytypical}, it can be shown that if $(\mathbf x, \mathbf y, \mathbf z_0)$ are jointly typical with respect to the $n$-letter random variables $(X^n,Y^n,Z^n)$, then $(\mathbf x, \mathbf y, \mathbf z_0)$ are also jointly typical with respect to the single-letter random variables $(X,Y, Z)$,   i.e.,
\begin{align*}
|P_{(\mathbf{x},\mathbf{y},\mathbf{z}_0)}(x ,y,z ) -P_{\mathbf x}(x)  p(y|x)p(z|x)| \leq \epsilon_3 P_{\mathbf x}(x)  p(y|x)p(z|x),
\end{align*}
for some $\epsilon_3 \to 0$ as $\epsilon_1 \to 0$. Then,
\begin{align*}
P_{( \mathbf{y},\mathbf{z}_0)}(0,1 )&\leq P_{\mathbf x}(0) (1-p)p (1+\epsilon_3)+P_{\mathbf x}(1)p (1-p)(1+\epsilon_3)\\
&=p(1-p)(1+\epsilon_3)
\end{align*}
and $P_{( \mathbf{y},\mathbf{z}_0)}(0,1 ) \geq p(1-p)(1-\epsilon_3)$. Similarly, we also have
\begin{align*}
P_{( \mathbf{y},\mathbf{z}_0)}(1,0 )\in [ p(1-p)(1-\epsilon_3) ,p(1-p)(1+\epsilon_3)  ],
\end{align*}
and thus
\begin{align*}
d(\mathbf{y},\mathbf{z}_0 )&=nBP_{( \mathbf{y},\mathbf{z}_0)}(0,1 )+nBP_{( \mathbf{y},\mathbf{z}_0)}(1,0 )\\
&\in [ 2nBp(1-p)(1-\epsilon_3) ,2nBp(1-p)(1+\epsilon_3)  ],
\end{align*}
i.e.,
\begin{align}
d(\mathbf{y},\mathbf{z}_0 )&\in [ nB(p*p-\epsilon_4) ,nB(p*p+\epsilon_4)  ],\label{E:bsca3}
\end{align}
where $\epsilon_4\to 0$  as $\epsilon_1\to 0$.

Combining \dref{E:bsca1}, \dref{E:bsca2} and \dref{E:bsca3}, we conclude that for any $\mathbf z_0 \in \mathcal T_{\epsilon}^{(B)}(Z^n)$ and some $\epsilon_0 \to 0$ as $\epsilon\to 0$,
\begin{align*}
\text{Pr}( p(\mathbf Y| \mathbf z_0)\leq 2^{-B(H(Y^n|Z^n)-\epsilon_0 )},d(\mathbf{Y},\mathbf{z}_0 )&\in [ nB(p*p-\epsilon_0) ,nB(p*p+\epsilon_0)  ]          |\mathbf z_0) \to 1 \text{ as } B\to \infty.
\end{align*}

\section{Property of $f(r)$}\label{A:f(r)}
For notational convenience, let $q:=p*p$. Taking the first derivative of $f(r)$, we have
\begin{align*}
f'(r)&=d_0 H'\left(\frac{r+d_0-q}{2d_0}\right)\cdot \left(\frac{r+d_0-q}{2d_0}\right)'
+(1-d_0)H'\left(\frac{r+q-d_0}{2(1-d_0)}\right) \cdot \left(\frac{r+q-d_0}{2(1-d_0)}\right)'\\
&=\frac{1}{2}\left[ \log \frac{d_0-r+q}{r+d_0-q}+\log \frac{2-d_0-r-q}{r+q-d_0} \right].
\end{align*}
With $r=d_0*q$, we have
\begin{align*}
f(d_0*q)&=d_0 H\left(\frac{d_0*q+d_0-q}{2d_0}\right)
+(1-d_0)H\left(\frac{d_0*q+q-d_0}{2(1-d_0)}\right) \\
&= d_0 H\left(\frac{(d_0(1-q)+(1-d_0)q)+d_0-q}{2d_0}\right)
+(1-d_0)H\left(\frac{(d_0(1-q)+(1-d_0)q)+q-d_0}{2(1-d_0)}\right) \\
&= d_0 H\left(1-q\right)
+(1-d_0)H\left(q\right) \\
&= H\left(q\right)
\end{align*}
and
\begin{align*}
f'(d_0*q)&=\frac{1}{2}\left[ \log \frac{d_0-d_0*q+q}{d_0*q+d_0-q}+\log \frac{2-d_0-d_0*q-q}{d_0*q+q-d_0} \right]\\
&=\frac{1}{2}  \log \left[\frac{d_0-(d_0(1-q)+(1-d_0)q)+q}{(d_0(1-q)+(1-d_0)q)+d_0-q}\cdot \frac{2-d_0-(d_0(1-q)+(1-d_0)q)-q}{(d_0(1-q)+(1-d_0)q)+q-d_0}\right]\\
&=\frac{1}{2}  \log\left[ \frac{2d_0q}{2d_0(1-q)}\cdot \frac{2(1-d_0)(1-q)}{2q(1-d_0)}\right]\\
&=0.
\end{align*}

Further taking the second derivative of $f(r)$ yields that
\begin{align*}
f''(r)
&=\frac{1}{2\ln 2}\left[  \frac{r+d_0-q}{d_0-r+q} \left( \frac{d_0-r+q}{r+d_0-q}\right)' + \frac{r+q-d_0}{2-d_0-r-q} \left( \frac{2-d_0-r-q}{r+q-d_0}\right)'\right]\\
&=\frac{1}{2\ln 2}\left[  \frac{2d_0}{(r+d_0-q)(r-d_0-q)} + \frac{2(d_0-1)}{(2-d_0-r-q)(r-d_0+q)}  \right].
\end{align*}
Provided the following constraint on $r$ (cf. \dref{E:rcondition1}--\dref{E:rcondition2})
\begin{align*}
r\in (\max\{q-d_0,d_0-q \}, \min\{q+d_0,2-q-d_0\}    )
\end{align*}
it can be easily seen that $f''(r)<0$. Therefore, $f(r)$ attains the maximum $H(p*p)$ if and only if $r=d_0*q=d_0*p*p$.

\end{document}